\documentclass[a4paper,UKenglish,cleveref, autoref, thm-restate]{lipics-v2021}

\pdfoutput=1 
\hideLIPIcs  

\title{A Congestion Parameter for Depth-First Graph Traversals}

\author{Codaline Bourotte
}{\'Ecole Normale Supérieure de Lyon, France}{coda.bourotte@ens-lyon.fr}{https://orcid.org/0009-0009-3488-0425}{Work supported in part by the grant ``Bourse Région Mobilité Internationale Etudiants de la région Auvergne-Rhône-Alpes'' and by Kubota Information Technology 
}

\author{Gwendal Ducloz}{\'Ecole Normale Sup\'erieure de Lyon, LIP (UMR 5668, Équipe MC2), France}{gwendal.ducloz@ens-lyon.fr}{https://orcid.org/0009-0002-2015-267X}{Work supported in part by the French government under the ``France 2030'' programme ANR-24-RRII-0001, ANR-22-CE09-0034, and ANR-22-CE45-0020.}

\author{Pekka Orponen}{Aalto University, Finland}{pekka.orponen@aalto.fi}{https://orcid.org/0000-0002-0417-2104}{}

\author{Shinnosuke Seki}{University of Electro-Communications, Japan}{s.seki@uec.ac.jp}{https://orcid.org/0000-0002-0276-3322}{Work supported in part by Ambition Internationale de la R\'egion Auvergne-Rh\^one-Alpes No.~00284230-1 and by the Aalto Science Institute visiting researcher programme.}

\authorrunning{C. Bourotte et al.} 

\Copyright{Codaline Bourotte, Gwendal Ducloz, Pekka Orponen, and Shinnosuke Seki} 



\category{} 

\relatedversion{} 

\acknowledgements{A major part of this work was done during an internship of C.B. at the University of Electro-Communications, Japan.
Further work was pursued during a visit of S.S. and G.D. to Aalto University, Finland, and during an internship of G.D. at the University of Electro-Communications.
}

\nolinenumbers 

\EventEditors{Michal Kouck\'{y} and Daniela Petri\c{s}an}
\EventNoEds{2}
\EventLongTitle{51st International Symposium on Mathematical Foundations of Computer Science (MFCS 2026)}
\EventShortTitle{MFCS 2026}
\EventAcronym{MFCS}
\EventYear{2026}
\EventDate{August 24--28, 2026}
\EventLocation{Paris, France}
\EventLogo{}
\SeriesVolume{386}
\ArticleNo{1}

\usepackage{algorithm}
\usepackage{algpseudocode}
\usepackage{bm}
\usepackage{comment}
\usepackage{nccmath}
\usepackage[textsize=tiny]{todonotes} 
\usepackage[normalem]{ulem}
\usepackage{xspace}

\usepackage{tikz}
\tikzset{vertex/.style={circle, draw, fill=black, inner sep=0pt, minimum size=4pt}}
\usetikzlibrary{shapes.geometric}

\newcommand{\KLX}{\ensuremath{\mathrm{KLX}}}
\newcommand{\TW}{\ensuremath{\mathrm{TW}}}

\newcommand{\MSO}{\ensuremath{\mathrm{MSO}}}
\newcommand{\dtc}{\ensuremath{\mathrm{dtc}}}
\newcommand{\DTC}{\ensuremath{\mathrm{DTC}}}

\newcommand{\gd}[1]{{\color{teal}#1}}
\newcommand{\gdsout}[1]{}

\newcommand{\parhead}[1]{\smallskip\noindent\textbf{#1. }\ignorespaces}

\begin{document}

\maketitle

\begin{abstract}
We explore a new graph parameter, \emph{KLX number}, which quantifies the minimum edge congestion of depth-first search (DFS) traversals of a given graph. 
Originally motivated by a problem in RNA nanostructure design, this parameter is also of independent theoretical interest.
Informally, the KLX number of a graph is defined as the minimum, over all its DFS traversals, of the maximum number of back edges that are simultaneously open during the traversal.

We provide full characterisations and linear-time recognition algorithms for graphs with KLX numbers 0, 1 and 2. We also relate KLX to tree-width, proving that any graph satisfies $\TW \le \KLX+1$. 
Furthermore, we show that the property $\KLX \le k$ is $\MSO_2$-expressible for every fixed $k$. 
Combined with the tree-width bound, this result implies
that determining whether a graph has KLX number at most $k$ can be achieved in linear time for any constant $k$.
\end{abstract}

\newpage

\section{Introduction}
\label{sec:introduction}

\subsection{Background}

In this paper we investigate the minimum ``congestion'' of depth-first graph traversals, in the sense of how many back edges are simultaneously open (discovered but not concluded) at any point during the traversal. This concept is analogous to the notion of \emph{spanning tree congestion} introduced by Ostrovskii~\cite{ostrovskii_minimal_2004}, but is specific to ordered DFS trees, and addresses the dynamic traversal process rather than the static DFS tree.

Our interest in this concept was spurred by a practical problem in RNA nanotechnology~\cite{stewart_rna-nanotechnology_2024}, viz.\ minimising the number of kissing-loop connector motifs in the algorithmic design of co-transcriptionally folding RNA wireframe nanostructures~\cite{elonen_algorithmic_2022,geary_single-stranded_2014,orponen_secondary_2025}.

\begin{figure}[ht]
  \centering
  
  \begin{subfigure}[t]{0.2\textwidth}
    \includegraphics[width=0.7\textwidth]{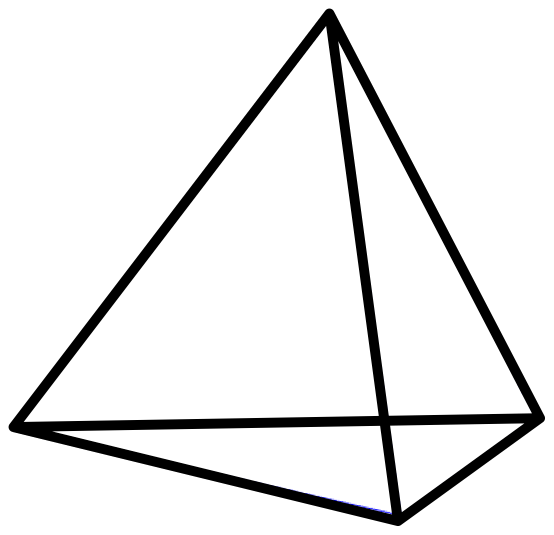}
    \subcaption{}
  \end{subfigure}
  \begin{subfigure}[t]{0.2\textwidth}
    \includegraphics[width=0.8\textwidth]{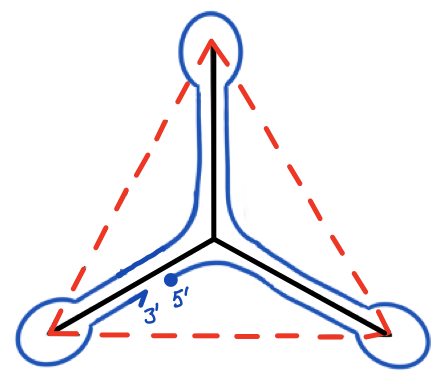}
    \subcaption{}
  \end{subfigure}
  \begin{subfigure}[t]{0.2\textwidth}
    \includegraphics[width=0.8\textwidth]{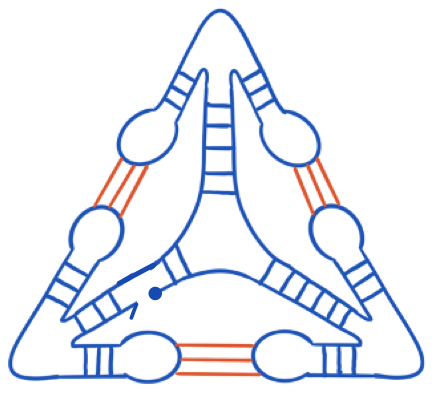}
    \subcaption{}
  \end{subfigure}
  \begin{subfigure}[t]{0.2\textwidth}
    \includegraphics[width=0.8\textwidth]{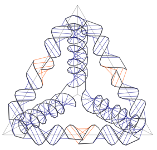}
    \subcaption{}
  \end{subfigure}
  
\caption{
A spanning-tree based design scheme for RNA wireframe nanostructures. (a) Targeted wireframe model. (b) A spanning tree and strand routing of the wireframe graph. (c) Routing-based helical and kissing-loop pairings. (d) Helix-level model. (Adapted with permission from~\cite{elonen_algorithmic_2022}.)
}
\label{fig:rna-design}
\end{figure}

In this design task, one starts with a graph model $G$ that represents the targeted wireframe nanostructure (\cref{fig:rna-design}(a)), and then designs an initial route for the RNA strand by a traversal around some spanning tree $T$ of $G$ (blue curve in \cref{fig:rna-design}(b)).
This establishes an RNA base-pairing scheme for creating the tree $T$, but omits the co-tree edges (dashed red lines in \cref{fig:rna-design}(b)).
These are grafted to the design by extending the strand routing by two half-edges from both end-vertices of each co-tree edge, and terminating these extensions in ``kissing loop'' (KL) pairs with pairwise matching base sequences (\cref{fig:rna-design}(c)).
\cref{fig:rna-design}(d) presents a helix-level model of the resulting nanostructure.


In nature, RNA is transcribed from a DNA template by a \emph{polymerase enzyme} molecule and folds \emph{co-transcriptionally}, that is, already while being transcribed. 
This opens up the possibility of reusing the KL base-pair sequences, because once a transcribed KL pair has closed, there is no longer a risk of mispairings from introducing another copy of the same sequence pair. 
This opportunity is potentially significant for large structures and motivates interest in cotranscriptional designs with a minimum number of simultaneously open KL pairs, that is, \emph{graph traversals with a minimum number of simultaneously open co-tree edges}.

Co-transcriptional folding also has a downside: if the scheduling of the KL pairings relative to the helix transcriptions is bad, it is possible that the polymerase molecule becomes ``trapped'' by a KL pair that closes too early, and the winding of the downstream helices cannot be completed~\cite{geary_design_2014}. Surprisingly, this kinetic folding trap can be decisively avoided in wireframe designs \emph{if and only if the spanning tree used for routing is a DFS tree}~\cite{orponen_secondary_2025}.

Accordingly, and following~\cite{orponen_secondary_2025}, we define the \KLX\ (``kissing loop crossing'') number of a graph~$G$ as the \emph{minimum number of simultaneously open back edges in any DFS traversal of~$G$}, and aim to characterise and recognise those graphs that can be constituted with a small \KLX\ number. Even though this parameter thus has an application-driven origin, it connects to the broad literature on edge congestion problems, e.g.~\cite{bodlaender_parameterized_2012, luu_better_2025, ostrovskii_minimal_2004}.

\subsection{Structure of the paper}

After introducing some preliminaries in \cref{sec:preliminaries}, we provide characterisations of graphs with small KLX values in \cref{sec:klx-leq-2}. Specifically, we show that:
graphs with $\KLX \leq 1$ are cactus graphs, and graphs with $\KLX \leq 2$ admit an explicit, linear-time decidable characterisation.
Computing the \KLX\ value is known to be NP-hard~\cite{orponen_secondary_2025}, but in \cref{sec:fpt-courcelle} we establish the existence of a linear-time fixed-parameter tractable (FPT) algorithm to test whether $\KLX \le k$ for any constant $k$. Our approach is based on Courcelle's Theorem~\cite{courcelle_monadic_1990}, and consists of two key steps: we prove that the tree-width $\TW$ of any graph satisfies $\TW \le \KLX +1$, and we demonstrate that the condition $\KLX \le k$ is $\MSO_2$-expressible. 
All the proofs that have been abridged or omitted in the main paper are provided in the Appendix. 

\section{Preliminaries}
\label{sec:preliminaries}

All graphs $G = (V, E)$ considered in this article are connected and simple. We denote by $\deg(x)$ the degree of a vertex $x \in V$, and by $\delta(G)$ and $\Delta(G)$ the minimum and maximum degrees of vertices in $G$.
For $S \subseteq V$, we denote by $G[S]$ the subgraph of $G$ induced by the vertices in $S$. 
We use the notation $|C|$ for the length of a sequence, path or cycle $C$.

A graph $G = (V, E)$ is \emph{$k$-(vertex-)connected} if $|V| > k$ and $G$ cannot be disconnected by removing fewer than $k$ vertices. A 2-connected graph is also called \emph{biconnected}.
By $\kappa(G)$, we denote the largest integer $k$ such that $G$ is $k$-connected. 
The \emph{$k$-edge-connectivity} and the notation $\lambda(G)$ are the edge analogues. 
By $k$-connectedness, we always mean $k$-vertex-connectedness. 
It is well known that $\kappa(G) \le \lambda(G) \le \delta(G)$. 
A \emph{biconnected component} or \emph{block} of $G$ is a maximal biconnected subgraph. Any graph decomposes into a tree of biconnected components, known as the \emph{block-cut tree} of the graph~\cite{bollobas_modern_1998}.

Starting from a vertex $r$ and traversing $G$ in a depth-first manner yields a spanning tree $T$ for $G$ of a special kind, called a \emph{DFS tree} or a \emph{Tr\'emaux tree}, rooted at~$r$. 
The unique path between two vertices $x$ and $y$ in $T$ is denoted by $xTy$, and the subtree rooted at $x$ is denoted by $T(x)$. 
The root $r$ induces a \emph{partial tree order}
$\prec$ on the vertices of $G$, where $u \prec v$ if $u \in rTv$ (equivalently if
$v \in T(u)$). The root is the unique minimum element of this order.

A vertex that is neither a leaf nor the root is called \emph{internal}, and an internal vertex of degree at least three in $T$ is called a \emph{branch vertex}.
We note the following simple lemma. 
\begin{lemma}\label{lem:root_DFST_biconnected-graph_non-branching}
    The root of a DFS tree of a biconnected graph is incident to only one tree edge and is thus not a branch vertex. 
\end{lemma}

The edges of $G$ are classified in two types: \emph{tree edges}, which belong to $T$, and \emph{back edges}, which do not. 
This nomenclature is justified by the property that, in a DFS tree, any back edge connects a vertex $v$ to one of its ancestors, i.e., a vertex in $rTv$. We consider back edges as oriented correspondingly and write $(v,u)$ for a back edge where $u \prec v$. In contrast, we use curly brackets for tree edges, such as $\{x,y\}$.

A DFS tree $T$ is \emph{ordered} if an ordering is specified for the children of each vertex. 
Each such ordering corresponds to a unique DFS traversal process that results in the tree $T$. 
Throughout the paper, when depicting an ordered tree, we assume that the children of a vertex are visited from left to right.

\label{sec:klx-number}

Given an ordered DFS tree $T$, its \emph{traversal} or \emph{(DFS-)contour} is the list of occurrences of its vertices on the tree tour. Formally, the traversal $\tau(T)$ of an ordered DFS tree $T$ with root $r$ is defined recursively as $\tau(r)$, where $\tau(u) = u$ if $u$ is a leaf, and $\tau(u) = u\, \tau(v_1)\, u\, \cdots u\, \tau(v_k)\, u$, if $u$ is a vertex with children $v_1,\dots,v_k$ ordered left to right. Note that in a traversal, every tree edge $\{x,y\}$ occurs twice, once as $xy$ and once as $yx$.

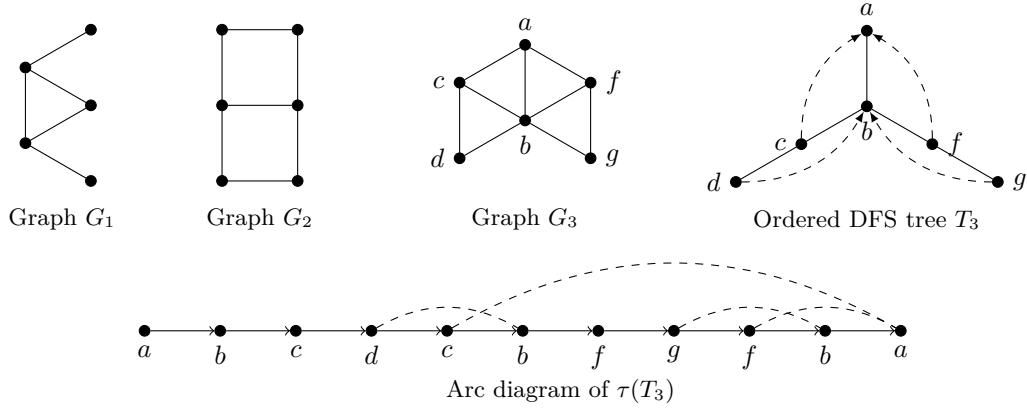
\begin{figure}
    \centering
    \begin{subfigure}{0.13\linewidth}
        \centering
        \begin{tikzpicture}
            \draw (0, 0) node[vertex] (a) []{}
            -- ++(-150:1)  node[vertex] (b) []{}
            -- ++(-30:1)  node[vertex] (c) []{}
            -- ++(-150:1)  node[vertex] (d) []{}
            -- ++(-30:1)  node[vertex] (e) []{}
            (b) -- (d)
            ;
            \draw (-0.38,-2.5) node {\small Graph $G_1$};

        \end{tikzpicture}
    \end{subfigure}
    \hspace{0.04\linewidth}
    \begin{subfigure}{0.13\linewidth}
        \centering
        \begin{tikzpicture}
            \draw (0, 0) node[vertex] (a) []{}
            -- ++(-90:1)  node[vertex] (b) []{}
            -- ++(-90:1)  node[vertex] (c) []{}
            -- ++(0:1)  node[vertex] (d) []{}
            -- ++(90:1)  node[vertex] (e) []{}
            -- ++(90:1)  node[vertex] (f) []{}
            (b) -- (e)
            (a) -- (f)
            ;
            \draw (0.5,-2.5) node {\small Graph $G_2$};
        \end{tikzpicture}
    \end{subfigure}
    \hfill
    \begin{subfigure}{0.2\linewidth}
        \centering
        \begin{tikzpicture}
                    \draw (0,0) node {~};

            \draw (0, -0.2) node[vertex] (a) [label=above:$a$]{}
            -- ++(270:1) node[vertex] (b) [label=below:$b$]{}
            -- ++(150:1) node[vertex] (c) [label=left:$c$]{}
            -- ++(270:1) node[vertex] (d) [label=left:$d$]{}
            (b)
            -- ++(30:1) node[vertex] (e) [label=right:$f$]{}
            -- ++(270:1) node[vertex] (g) [label=right:$g$]{}
            (a) -- (c)
            (a) -- (e)
            (b) -- (d)
            (b) -- (g)
            ;
         \draw (0,-2.5) node {\small Graph $G_3$};
        \end{tikzpicture}
    \end{subfigure}
    \hspace{0.03\linewidth}
    \begin{subfigure}{0.35\linewidth}
        \centering
        \begin{tikzpicture}
             \draw (0, 0) node[vertex] (a) [label=above:$a$]{}
            -- ++(270:1) node[vertex] (b) [label=below:$b$]{}
            -- ++(210:1) node[vertex] (c) [label=left:$c$]{}
            -- ++(210:1) node[vertex] (d) [label=left:$d$]{}
            (b)
            -- ++(330:1) node[vertex] (e) [label=right:$f$]{}
            -- ++(330:1) node[vertex] (g) [label=right:$g$]{}
            ;
            \draw[dashed, -latex]
            (c) to [out=90, in=210] (a)
            ;
            \draw[dashed, -latex]
            (e) to [out=90, in=330] (a);
            \draw[dashed, -latex]
            (d) to [out=0, in=240] (b)
            ;
            \draw[dashed, -latex]
            (g) to [out=180, in=300] (b)
            ;
            \draw (0,-2.5) node {\small Ordered DFS tree $T_3$};

        \end{tikzpicture}
    \end{subfigure}
    \begin{tikzpicture}
    \node[vertex, label=below:$a$] (a) at (0, 0) {};
    \node[vertex, label=below:$b$] (b) at (1, 0) {};
    \node[vertex, label=below:$c$] (c) at (2, 0) {};
    \node[vertex, label=below:$d$] (d) at (3, 0) {};
    \node[vertex, label=below:$c$] (c2) at (4, 0) {};
    \node[vertex, label=below:$b$] (b2) at (5, 0) {};
    \node[vertex, label=below:$f$] (e) at (6, 0) {};
    \node[vertex, label=below:$g$] (g) at (7, 0) {};
    \node[vertex, label=below:$f$] (e2) at (8, 0) {};
    \node[vertex, label=below:$b$] (b3) at (9, 0) {};
    \node[vertex, label=below:$a$] (a2) at (10, 0) {};

    \draw[->] (a) -- (b);
    \draw[->] (b) -- (c);
    \draw[->] (c) -- (d);
    \draw[->] (d) -- (c2);
    \draw[->] (c2) -- (b2);
    \draw[->] (b2) -- (e);
    \draw[->] (e) -- (g);
    \draw[->] (g) -- (e2);
    \draw[->] (e2) -- (b3);
    \draw[->] (b3) -- (a2);

    \path[dashed]
        (c2) edge[out=30, in=150, -] (a2)
        (e2) edge[out=30, in=150, -] (a2)
        (d) edge[out=30, in=150, -] (b2)
        (g) edge[out=30, in=150, -] (b3)
        ;
    \draw (5.5,-0.8) node {\small Arc diagram of $\tau(T_3)$};
\end{tikzpicture}
    \caption{Three graphs $G_1$, $G_2$, and $G_3$, along with an ordered DFS tree $T_3$ of $G_3$ rooted at $a$ and the corresponding arc diagram.
    }
    \label{fig:congestion_differs_from_KLX}
\end{figure}

Given a DFS traversal $\tau(T)$, the corresponding \emph{arc diagram} is constructed by laying out all the vertex occurrences of $\tau(T)$ in an oriented line, and representing each back edge $(v,u)$ as an arc from the last occurrence of $v$ in $\tau(T)$ to the next occurrence of $u$ after the last occurrence of $v$---which exists as $u \prec v$ in $T$. An arc diagram is illustrated in \cref{fig:congestion_differs_from_KLX}.

\begin{definition}[KLX]
Let $e=\{x,y\}$ be a tree edge of an ordered DFS tree $T$, with $x \prec y$. The set $\mathcal{O}(e)$ of \textbf{open (back) edges for e} is defined as the set of back edges whose span in the corresponding arc diagram contains the oriented edge $y \to x$.


We define $\KLX(T) = \KLX(G,T) = \max_{e \in E_T} |\mathcal{O}(e)|$ and $\KLX(G) = \min_T \KLX(G,T)$, where the minimisation is over all ordered DFS trees $T$ of $G$.
An (ordered) DFS tree $T^*$ that satisfies $\KLX(G,T^*) = \KLX(G)$ is  a \textbf{KLX-minimum tree} of $G$.
\end{definition}

For instance, in the arc diagram of \cref{fig:congestion_differs_from_KLX}, we have $\mathcal{O}(\{b,f\}) = \{(c,a),(g,b), (f,a)\}$ and $\mathcal{O}(\{c,d\}) = \{(d,b)\}$. In this case, $\KLX(T_3) = 3$.

\begin{remark}
A graph has $\KLX = 0$ if and only if it is a tree, as in this case there are no back edges.
\end{remark}

\begin{example}\label{example1}
Figure~\ref{fig:congestion_differs_from_KLX} presents three graphs $G_1$, $G_2$, and $G_3$, with KLX numbers 1, 2, and 3, respectively.  
The KLX values can be verified by full enumerations of the respective DFS trees, but the present paper also provides simpler techniques.
Graph $G_1$ is not a tree but a ``cactus,'' and hence has $\KLX = 1$ (Theorem~\ref{thm:characterisation_of_KLX=1}). Graph $G_2$ admits a path-like DFS tree that fits a characterisation of $\KLX = 2$ graphs in Theorem~\ref{th:biconnect_KLX2}. 
Graph $G_3$ does not satisfy the conditions of Theorem~\ref{th:biconnect_KLX2}, and thus has $\KLX(G_3) \ge 3$. 
Since we have just observed that $\KLX(G_3, T_3) = 3$, we can conclude that in fact $\KLX(G_3) = 3$. 
\end{example}

The back edges in $\mathcal{O}(e)$ can be classified in two types, based on their relation to the tree edge $e$. 
We say that a back edge $(v,u)$:
 \begin{enumerate}[(i)]
     \item \emph{crosses} the edge $e$ if $e \in uTv$; 
     \item \emph{envelops} the edge $e$ if there exists a branch vertex $x \in uTv$ with $x \neq u, v$, and with children $w_1, w_2$ such that $w_1$ precedes $w_2$, $v \in T(w_1)$, and either $e \in T(w_2)$ or $e = \{x,w_2\}$. 
 \end{enumerate}
Similarly, we say that a back edge $(v,u)$ \emph{crosses a vertex $x$} if $x \in uTv$ and $x \neq u, v$. 

For example, in the ordered DFS tree $T_3$ and its arc diagram in \cref{fig:congestion_differs_from_KLX}, among the three back edges in $\mathcal{O}(\{b,f\})$, the edges $(g,b)$ and $(f,a)$ are crossing, and the edge $(c,a)$ is enveloping.

This classification of back edges suggests the following simple lower bound for the KLX number, obtained by taking into account only the crossing edges.

\begin{definition}[DFS tree congestion]
    For a tree edge $e$ in a DFS tree $T$ of a graph $G$, we define its \textbf{(DFS-tree) congestion} $\dtc(e)$ as the number of back edges $(v, u)$ that cross $e$.
    The congestion of a DFS tree $T$ is then defined as $\DTC(T) = \DTC(G,T) = \max_{e \in T} \dtc(e)$, and that of a graph $G$ as $\DTC(G) = \min_T \DTC(G, T)$.
\end{definition}

We observe that the value of $\DTC(T)$ is independent of the traversal order of the tree $T$, and thus the quantity $\DTC(G)$ matches Ostrovskii's spanning tree congestion parameter \textit{stc}~\cite{ostrovskii_minimal_2004}, when the latter is restricted to DFS trees and adjusted by $-1$.


\begin{proposition} \label{prop:klx-vs-dtc}
    For any graph $G$, $\DTC(G) \le \KLX(G)$. 
    Moreover, there exist graphs for which this inequality is strict. 
\end{proposition}

The strictness of the lower bound is demonstrated e.g.\ by the biconnected graph $G_3$ in \cref{fig:congestion_differs_from_KLX}, where from \cref{example1} we know that $\KLX(G_3) = 3$, and DFS tree $T_3$ witnesses that $\DTC(G) \leq 2$.




\section{Characterisations of graphs with $\bf \textbf{KLX} \leq 1$ and $\bf \textbf{KLX} \leq 2$}
\label{sec:klx-leq-2}

In this section we provide characterisations of graphs with $\KLX$ number less than or equal to two.
We begin by establishing some foundational lemmas, and then build on these to obtain the targeted characterisations.
Finally, we present linear-time algorithms to decide whether a given graph satisfies $\KLX \leq 2$.  

\subsection{Some preliminary lemmas}
\begin{lemma}\label{lem:KLX_of_a_block_is_atmost_that_of_the_whole}
    A block (biconnected component) $B$ of a graph $G$ satisfies $\KLX(B) \le \KLX(G)$.  
\end{lemma}

\begin{proof}
    Let $T$ be a KLX-minimum tree of $G$, and let $B$ be a block of $G$. 
    Consider the subtree $T[V_B]$ of $T$ induced by the vertices of $B$.
    Since $B$ is biconnected, all paths in $G$ between vertices of $B$ lie entirely within $B$.
    Consequently,  $T[V_B]$ is connected and acyclic, and thus forms a spanning tree (and hence an ordered DFS tree) of $B$.
    By the definition of KLX, the inequality 
    \(\KLX(B) \le \KLX(B, T[V_B]) \le \KLX(G)\) holds.
\end{proof}

\begin{lemma}\label{lem:num_suspenders<=vertex-connectivity-1}
    Let $k \ge 2$, and let $T$ be a DFS tree of a $k$-connected graph, rooted at a vertex~$r$.  
    Any subtree of $T$ rooted at a child of a vertex $v$ is then connected to the path $rTv$ by at least $k-1$ back edges.
    Moreover, unless $v = r$, at least one back edge from each subtree crosses the tree edge between $v$ and its parent. 
\end{lemma}
\begin{proof}
    Suppose, for a contradiction, that for some subtree of $T$ rooted at a child of $v$ there are at most $k{-}2$ such back edges.
    Cutting all these back edges, along with the tree edge between $v$ and the respective child, would disconnect $G$.
    This would imply
    $\kappa(G) \le \lambda(G) \le k-1$, contradicting the $k$-connectivity of $G$. 

    Regarding the second statement, without such a crossing back edge, removing vertex $v$ would disconnect the subtree from the rest of $G$, contradicting the biconnectivity of $G$.
\end{proof}


\begin{corollary}\label{cor:KLX_upperbounds_branch-size}
    In a DFS tree $T$ of a biconnected graph $G$ every vertex has at most $\DTC(T) \leq \KLX(T)$ children.
\end{corollary}


\begin{restatable}{lemma}{KLXminDFStreedegreeupperbound}
    \label{lem:KLX-min-DFS-tree_degree_upperbound}
    Let $G$ be a biconnected graph and $T$ a DFS tree of $G$. 
    Then any vertex of $G$ has degree at most $\KLX(T)+1$ if it is the root or a leaf of $T$, and degree at most $2\cdot\KLX(T)+1$ if it is an internal and not a branch vertex of $T$. 
\end{restatable}

\begin{proof}
See Appendix~\ref{sec:appendix-a}, \cref{lem:proof_lems_sec_3}.
\end{proof}


\begin{restatable}{lemma}{maxdegreeandKLXlowerbound}
    \label{lem:max-degree_and_KLX_lowerbound}
    Let $G$ be a biconnected graph with 
    $\Delta(G) \ge \frac{k(k+1)}{2}+2$ for some $k \ge 1$. Then $\KLX(G) \ge k+1$. 
\end{restatable}

\begin{proof}
    Suppose, for a contradiction, that there is a DFS tree $T$ of $G$ such that $\KLX(T) \le k$. 
    According to \cref{lem:KLX-min-DFS-tree_degree_upperbound}, any vertex that is the root or a leaf of $T$ has degree at most $k+1$, and any internal vertex that is not a branch vertex of $T$ has degree at most $2k+1$. 
    It thus now suffices to show that any other (branch) vertex of $T$ has degree at most $\frac{k(k+1)}{2}+1$ to get a contradiction with the assumption on $\Delta(G)$. 

    By \cref{cor:KLX_upperbounds_branch-size} the number of children of any branch vertex of $T$ is bounded by $\KLX(T)$. 
    Let $b$ be a branch vertex with $p \le k$ children. 
    Since $T$ cannot branch at its root (\cref{lem:root_DFST_biconnected-graph_non-branching}), $b$ has a parent. Let $e \in T$ be the edge connecting $b$ to its parent, and
    let $e_1, \cdots, e_p \in T$ be the edges connecting $b$ with its children. Without loss of generality, assume that the traversal of $T$ explores $e_1$ first, $e_2$ second, and so on.
    
    The remaining edges incident to $b$, if any, are back edges and must cross $e$, $e_1$, $\dots$, or $e_p$.  
    As $G$ is biconnected, for all $1 \le i \le p$, there exists a back edge spanning from the subtree below the $i$-th child to crossing $e$ (\cref{lem:num_suspenders<=vertex-connectivity-1}). 
    Moreover, this $i$-th back edge envelops the edges $e_{i+1}, \cdots, e_p$ and all edges below them. 
    Therefore, $e_i$ and $e$ can be crossed by at most $k-i$ and $k-p$ more back edges, respectively. 
    Hence, the degree of the branch vertex $b$ is at most 
    \begin{eqnarray*}
        p+1 + k-p + \sum_{i=1}^{p} (k-i) &=& -\frac{1}{2}\left( p- \left(k-\frac{1}{2}\right)\right)^2 +\frac{k(k+1)}{2}+\frac{9}{8} \le \frac{k(k+1)}{2}+1
    \end{eqnarray*}
\end{proof}

\subsection{Characterisation of graphs with $\bf \textbf{KLX} \leq 1$}

A graph $G$ is a \emph{cactus graph} if every block of $G$ is either an edge or a cycle~\cite{araujo_proper_2016}.

\begin{theorem}\label{thm:characterisation_of_KLX=1}
    A graph $G$ has $\KLX(G) \le 1$ if and only it is a cactus graph. 
\end{theorem}

\begin{proof}

For the only-if direction,  \cref{lem:KLX_of_a_block_is_atmost_that_of_the_whole} implies that every block $B$ of $G$ satisfies $\KLX(B) \le 1$. Applying \cref{lem:max-degree_and_KLX_lowerbound} with $k=1$, no block involves a vertex with degree 3 or greater. 
Therefore, each block $B$ must be either an edge or a simple cycle. 
Thus, $G$ is a cactus graph.

For the other direction, given a cactus graph $G$, we construct a DFS traversal of $G$ in which any tree edge is crossed by at most one back edge and enveloping never occurs. 
The traversal can start at an arbitrary vertex.
After the traversal enters a cycle block $C$ via some articulation point $u$, it prioritises visiting the neighbour blocks of $C$ (if any) in the block-cut tree before continuing along the cycle $C$.
By this strategy, the only back edge ever open in $C$ is the one connecting the final vertex of $C$ back to $u$, which closes when the traversal returns back to $u$. Thus, at most one back edge in $G$ is open at any time, and hence $\KLX(G) \leq 1$.
\end{proof}

\begin{corollary}
There is an $O(|V| + |E|)$ algorithm to test whether a graph has $\KLX \leq 1$.
\end{corollary}

    \subsection{Characterisation of graphs with $\bf \textbf{KLX} \leq 2$}
In the following, we first provide a structural characterisation of biconnected graphs with $\KLX = 2$, and then generalise this result to general graphs. 

\begin{lemma}\label{lem:back-paths_of_KLX=2_DFST_biconnected-graph}
    Let $T$ be a DFS tree of a biconnected graph $G$ with $\KLX(T)\leq2$. 
    Then any branch vertex $b$ has only two children in $T$, and at least one of the subtrees rooted at them has the following properties: 
    \begin{enumerate}
        \item The subtree is a path (contains no further branch vertices).
        \item All of the subtree is crossed by a single back edge that connects its (unique) leaf to an ancestor of $b$. No other back edges are incident to the subtree.
    \end{enumerate} 
\end{lemma}
\begin{proof}
    Firstly, $b$ is not the root of $T$ (\cref{lem:root_DFST_biconnected-graph_non-branching}) and has only two children (\cref{cor:KLX_upperbounds_branch-size}). 
    Let $T_1,\ T_2$ denote the two subtrees rooted at these children and assume, without loss of generality, that in the ordering of $T$ that yields $\KLX(T) \leq 2$, tree $T_1$ precedes tree $T_2$. 
    
    By \cref{lem:num_suspenders<=vertex-connectivity-1}, there exists at least one back edge from $T_1$ to an ancestor of $b$ that envelops all edges in $T_2$.
    This entails that no edge in $T_2$ can be crossed by more than one back edge. 
    As a result, $T_2$ can not contain a branch vertex, as the edge from the branch vertex to its parent would be crossed by two back edges (\cref{lem:num_suspenders<=vertex-connectivity-1}). 
    Consequently, $T_2$ must be a path.

    Secondly, the unique leaf $\ell$ of $T_2$ must be incident to a back edge lest its parent be an articulation point. 
    And the other endpoint $w$ of this back edge must lie above the branch vertex $b$; otherwise, to prevent $w$ from being an articulation point, another back edge would need to cross $w$. But then the tree edge connecting $w$ to its only child would be crossed by two back edges, contradicting our earlier conclusion.
    Consequently, the back edge from $\ell$ to $w$ crosses all edges in $T_2$, and no other back edges are incident to $T_2$. 
\end{proof}

\noindent\textbf{A new perspective: a spanning path $\bf \tilde{T}$ with back paths.}
 Let us note that a biconnected graph with $\KLX = 1$ is a cycle, by \cref{thm:characterisation_of_KLX=1}. From now on, we exclude this case, and focus on characterising biconnected graphs with KLX exactly equal to 2.

Let $G$ be such a graph and $T$ a KLX-minimum tree of $G$. By \cref{lem:back-paths_of_KLX=2_DFST_biconnected-graph}, every branch vertex $b$ in $T$ has at least one path-like subtree $b \rightarrow \cdots \rightarrow \ell$, with a unique back edge $(\ell,v)$ that crosses the whole subtree. We refer to such a path $b \rightarrow \cdots \rightarrow \ell \rightarrow v$ as a \emph{back path} from $b$ to $v$, and denote it by $[b,v]$. At the lowermost branch vertex — which exists as the graph is biconnected and not a cycle — both subtrees become back paths, and we consider also the other back edges, that do not participate in back paths, as special cases of this concept.

Contracting the path-like subtrees to back paths yields now a reduced spanning tree, in fact a spanning path $\tilde{T}$. The notions of edge crossings and the $\KLX$ number can be extended to $\tilde{T}$, resulting in $\KLX(G,\tilde{T}) = \KLX(G,T)$.
Note also that there are now no enveloping back paths, as the tree $\tilde{T}$ no longer contains branch vertices.

\medskip

\begin{figure}
    \centering
    \begin{tikzpicture}[scale=1, transform shape]
        \draw[very thick] (0, 0) node[vertex] (root) [label=left:$r$] {}
        -- ++(0:0.8) node[vertex] (i1) {}
        -- ++(0:0.4) node[vertex] (k2) {}
        -- ++(0:0.8) node[vertex] (i2) {}
        -- ++(0:0) node[vertex] (k3) {} 
        -- ++(0:0.8) node[vertex] (i3) {}
        -- ++(0:0.4) node[vertex] (v1) [label={[label distance=-2pt]above:$l_2$}] {}
        -- ++(0:0.8) node[vertex] (b1) [label={[label distance=-2pt]below:$l'_1$}] {}
        -- ++(0:0.4) node[vertex] (k4)  {}
        -- ++(0:0.8) node[vertex] (i4) {}
        -- ++(0:0.4) node[vertex] (v2) [ label={[label distance=-2pt]below:$l_3$}] {}
        -- ++(0:0.8) node[vertex] (b2) {}
        -- ++(0:0.8) node[vertex] (b3) [label={[label distance=-2pt]below:$l'_3$}] {}
        -- ++(0:0.4) node[vertex] (k5) {}
        -- ++(0:0.8) node[vertex] (b4) [label={[label distance=-2pt]above:$l'_4$}] {}
        ;

        \node[above,yshift=-2pt] at (b2) {$l'_2=l_4$}; 

        \draw[dotted, thick] 
        (root) to [out=45, in=135] (i1)
        (k2) to [out=45, in=135] (i2)
        (k3) to [out=45, in=135] (i3)
        (k4) to [out=315, in=225] (i4)
        (k5) to [out=315, in=225] (b4)
        ;
        
        \draw[dashed, thick] 
        (root) to [out=60, in=120] (b1)
        (v1) to [out=300, in=240] (b2)
        (v2) to [out=75, in=105] (b3)
        (b2) to [out=300, in=240] (b4)
        ;
    \end{tikzpicture}
    \hfill
    \begin{tikzpicture}[scale=1, transform shape]
        \draw node[vertex] (root) at (0,-0.5) [label=left:$r$] {};
        \draw node[vertex] (v1) at (1,-1) [label=below:$l_2$] {};
         \draw[very thick] (0, 0) (root) -- (v1);         
         \draw[very thick, color=red] (v1)
         -- ++(90:1) node[vertex, color=black] (b1) [label=above:\color{black}$l'_1$] {};
         \draw[very thick, color=black] (b1)
         -- ++(0:1) node[vertex, color=black] (v2) [label=above:$l_3$] {};
         \draw[very thick, color=red] (v2)
         -- ++(-90:1) node[vertex, color=black] (b2) [label=below:\color{black}$l'_2$] {};
         \draw[dashed, thick] (v2) -- ++(0:1) node[vertex] (b3) [label=above:$l'_3$] {};
         \draw[dashed, thick] (root) -- (b1);
         \draw[very thick, color=red] (b2) -- (b3) {};
         \draw[very thick, color=black] (b3)
         -- ++(-90:1) node[vertex] (b4) [label=below:$l'_4$] {};
         \draw[dashed, thick] (b2) -- (v1);
         \draw[dashed, thick] (b2) -- (b4);
    \end{tikzpicture}
    \caption{
    (Left) A DFS tree of a biconnected graph with $\KLX=2$, consisting of four long back paths $[r, l'_1], [l_2, l'_2], [l_3, l'_3], [l_4, l'_4]$ (drawn in dashed lines) and short back paths (dotted lines).
    (Right) A structural version of the corresponding graph, where all edges are paths, and black edges may contain cactus chains.}
    \label{fig:chain_of_4_cactus_necklaces}
\end{figure}
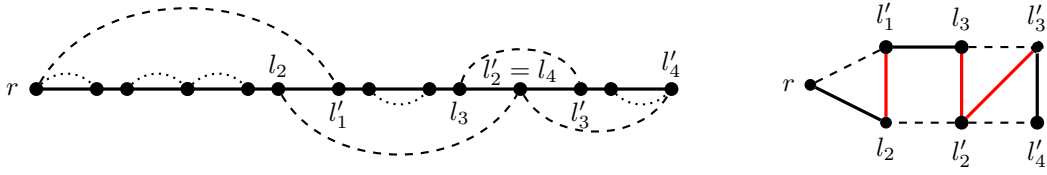

\parhead{Long back paths}
We first remark that, as $\KLX(\tilde{T}) \le 2$, no tree edge of $\tilde{T}$ can be crossed by more than two back paths.
Consequently, the root $r$ is incident to at most two back paths. 
Let us then define a notion of \textit{long back paths} by induction.\\

\noindent\emph{Initialisation.}  The first long back path $[l'_1, l_1]$ is defined as follows:
\begin{enumerate}[(i)]
    \item If $r$ is incident to only one back path $[ u,r]$, we set $[l'_1, l_1] = [u,r]$.
    \item If $r$ is incident to two back paths $[ u_1,r]$ and $[u_2,r]$, with $u_1 \preceq u_2$, we set $[l'_1, l_1] = [u_2,r]$.
\end{enumerate}

\noindent\emph{Induction.} 
 Let $[l'_i, l_i]$ be lastly defined long back path, with $l_i \prec l'_i$.
\begin{enumerate}[(i)]
    \item If $l'_i$ is the single leaf of $\tilde{T}$, then all long back paths have been defined.
    \item Otherwise, 
    to avoid being an articulation point, the branch vertex $l'$ must be crossed by a back path, but no more than one, as $\KLX(\tilde{T}) \le 2$. 
    This unique back path $[l'_{i+1},l_{i+1}]$ is defined as the next long back path. 
    Moreover, $l_i \prec l_{i+1} \prec l'_i \prec l'_{i+1}$. 
    For $i=1$, it comes from the definition of the first long back path as the ``longest'' one. For $i\geq2$, the relation holds as $[l'_i, l_i]$ is the only back path to cross $l'_{i-1}$, hence, $l_i \prec l'_{i-1} \preceq l_{i+1} $
\end{enumerate}

\noindent All long back paths, hence, satisfy the following relation: 
\begin{ceqn}
\begin{equation}\label{eq:KLX2_braids}
    r = l_1 \prec l_2 \prec l'_1 \preceq l_3 \prec l'_2 \preceq l_4 \prec\cdots \prec l'_{p-2} \preceq l_p \prec l'_{p-1} \prec l'_p, 
\end{equation}
\end{ceqn}
Note that if $\Delta(G) = 4$, then $l_{i+2}$ can be equal to $l'_{i}$ for all $1 \le i \le p-2$ (for instance, in \cref{fig:chain_of_4_cactus_necklaces}, $l'_2 = l_4$). Note however, that $\Delta(G) < 5$ according to \cref{lem:max-degree_and_KLX_lowerbound}.

\medskip

\parhead{Short back paths}
Let us now classify the remaining back paths, called \textit{short back paths}.
By construction, for each $1\leq i \leq p-1$, the tree path between $l_{i+1}$ and $l'_i$ is crossed by two long back paths, hence, no other back path can cross one of these edges.

However, for each $2 \le i \le p-1$, if $l_{i+1} \neq l'_{i-1}$, all edges along the tree path $l'_{i-1} \tilde{T} l_{i+1}$ are crossed only by the long back path $[l'_{i}, l_{i}]$, allowing one more back path to cross them. 
By characterisation of graphs with KLX 1, the only remaining possibility at this location is arbitrarily many \emph{short back paths} $[k_{i, 1}', k_{i, 1}], \ldots, [k'_{i, m}, k_{i, m}]$ ($m \ge 0$), aligned in series, as
\begin{ceqn}
\begin{equation}\label{eq:KLX2_short_bps}
    l_i \prec l'_{i-1} \preceq k_{i, 1} \prec k_{i, 1}' \preceq k_{i, 2} \prec k_{i, 2}' \preceq \cdots \preceq k_{i, m} \prec k_{i, m}' \preceq l_{i+1} \prec l'_i 
\end{equation}
\end{ceqn}

Finally, the only non-treated location are on the tree path between $l_1$ and $l_2$ and between $l'_{p-1}$ and $l'_{p}$. Similarly, the only possibility at these locations are arbitrarily many short back paths, aligned in series. We spot these short back paths with indices 1 and p, respectively.
\begin{ceqn}
\begin{align}\label{eq:KLX2_short_bps_first}
    &\text{If $p = 1$:\hspace{3.2em}}  l_1 \preceq k_{1, 1} \prec k_{1, 1}' \preceq k_{1, 2} \prec k_{1, 2}' \preceq \cdots \preceq k_{1, m} \prec k_{1, m}' \preceq l_{1}' \\
    &\text{If $p \geq 2$:\hspace{1em}} \left\{ 
    \begin{aligned} 
    l_1 \preceq k_{1, 1} \prec k_{1, 1}' \preceq k_{1, 2} \prec k_{1, 2}' \preceq \cdots \preceq k_{1, m} \prec k_{1, m}' \preceq l_{2}  \\
    l'_{p-1} \preceq k_{p, 1} \prec k_{p, 1}' \preceq k_{1, 2} \prec k_{p, 2}' \preceq \cdots \preceq k_{p, m} \prec k_{p, m}' \preceq l'_{p} \nonumber
    \end{aligned}
    \right.
\end{align}
\end{ceqn}

\begin{theorem}[Biconnected graph with $\KLX =2$]  \label{th:biconnect_KLX2}
    A biconnected graph $G$ has $\KLX =2$ if and only if it admits a DFS path $\tilde{T}$, and a partition of its back paths into long and short ones such that
    \begin{enumerate}[(i)]
        \item long back paths satisfy equation (\ref{eq:KLX2_braids}), and
        \item short back paths satisfy equations (\ref{eq:KLX2_short_bps}) and (\ref{eq:KLX2_short_bps_first}).
    \end{enumerate}
\end{theorem}

\begin{proof}
The previous paragraphs prove the forward direction.
In the opposite direction, a graph with a DFS path rooted at $r$, with long and short back paths satisfying all previous equations can be traversed in a way that no edge is enveloped or crossed by two back edges. 
Starting from the root $r$, enveloping can be avoided by first going straight down to the leaf $b_n$ and then visiting the back paths during the back traversal, as stated in the proof of \cref{lem:back-paths_of_KLX=2_DFST_biconnected-graph}. 
For example, \cref{fig:chain_of_4_cactus_necklaces} illustrates that no tree edge is crossed by more than two back paths. 
Thus, this traversal results in a DFS tree with $\KLX=2$. 
Together, all arguments from this subsection prove \cref{th:biconnect_KLX2}.
\end{proof}
    
\begin{restatable}{corollary}{corcharKLX}
    \label{cor:charKLX2}
    A graph $G$ has $\KLX \le 2$ if and only if it admits a DFS tree $T$ with the following structural properties:
    \begin{enumerate}
    \item For each biconnected component $C$ of $G$, the restriction of $T$ to $C$ is a DFS tree $T_C$ exhibiting the back-path structure characterised in \cref{th:biconnect_KLX2}.
    \item The articulation points linking each component $C$ to the rest of the graph satisfy the following:        \begin{itemize}
        \item Components with $\KLX = 2$ are articulated only along the main tree path or along one of the last (either long or short) back paths of the corresponding DFS tree $T_C$.
        \item All other articulation points connect exclusively to cactus subgraphs.
        \end{itemize}
    \end{enumerate}
    
\end{restatable}

\begin{proof}[Proof overview.]
We consider a KLX-minimum DFS tree and restrict it to each biconnected component $C$. The possible locations of articulation points then follow from a case analysis based on the structural characterisation given in \cref{th:biconnect_KLX2}.
\end{proof}

    \subsection{Algorithms for recognising graphs with $\bf \textbf{KLX} \leq 2$}

For an algorithm recognising graphs with $ \text{KLX} \leq 2$, we first need three structural definitions:
    \begin{enumerate}[(i)]
        \item A \textit{2-high-degree cycle} is a simple cordless cycle containing exactly two vertices of degree greater than 2.
        \item A \textit{ladder pattern} is a sequence of simple cycles such that:
        \begin{itemize}
            \item No cycle contains a vertex of degree 5 or higher.
            \item Any two adjacent cycles in the sequence share exactly one common path.
            \item Non-adjacent cycles in the sequence share no common path.
        \end{itemize}
        \item Let $\mathcal{C}_1, \ldots, \mathcal{C}_p$ (with $p \geq 2$) be a ladder pattern, and let $(u_1,d_1), \ldots,(u_{p-1},d_{d-1})$ be the paths common to each pair of neighbouring cycles, such that all $u_i$ (respectively all $d_i$) lie on the same side of the ladder. This means that all paths from $u_i$ to $u_{i+1}$ belong only to $\mathcal{C}_{i+1}$.
        A \textit{zigzag motif of the ladder pattern} is then a sequence of edges of the form:
        \begin{itemize}
            \item $\{(v,d_1),(u_1,u_2), (d_2, d_3), (u_3, u_4), \ldots, (\cdot,w)\}$, or
            \item $\{(v,u_1), (d_1,d_2), (u_2, u_3), (d_3, d_4), \ldots, (\cdot,w)\}$,
        \end{itemize}
         where $v$ and $w$ are vertices located on one of the paths composed of degree 2 vertices, either between $u_1$ and $d_1$ or between $u_p$ and $d_p$, respectively.
    \end{enumerate}


 \begin{algorithm}[htbp]
\caption{An algorithm recognising biconnected graphs with $ \text{KLX} \leq 2$}
\label{alg:KLX2}
\begin{algorithmic}[1]
\State Check that there is no vertex of degree at least 5, and test whether $\KLX(G)\leq 1$. \label{step:initial}
\State Replace every 2-high-degree cycle by a path of length~2 between the two high degree vertices. Proceed linearly, without using the ``replaced'' edges.
\label{step:ignore_short}
\State Verify that the resulting graph forms a ladder pattern.\label{step:braid_detecting}
\State Check the placement of the replaced cycles:
\begin{itemize}
    \item If the ladder pattern consists of a single cycle, accept if there is a vertex of degree 3.
    \item Otherwise, verify that there exists a zigzag motif of the ladder pattern and that all replaced cycles are correctly positioned along this zigzag motif.
\end{itemize}
 \label{step:final_check}
\end{algorithmic}
\end{algorithm}

 \begin{restatable}{proposition}{algbicKLX}
 \label{prop:algoKLX2}
    A biconnected graph $G$ passes all the steps of \cref{alg:KLX2} if and only if $\KLX(G) \leq 2$. \Cref{alg:KLX2} has time complexity $\mathcal{O}(|V| + |E|)$. 
\end{restatable}

\begin{proof}[Proof overview.]
    We first show that 2-high-degree cycles correspond to short back paths. Once these are ignored, the graph  must contain only long back paths according to \cref{th:biconnect_KLX2}, which in turn form exactly a ladder pattern.
    
    The remaining task is to ensure that the short back paths are properly located along the DFS path.
    If $G$ passes all steps of \cref{alg:KLX2}, we can reconstruct a DFS tree that satisfies \cref{th:biconnect_KLX2}. This DFS traversal follows the zigzag motif of the identified ladder pattern.
    
    All steps run in linear time, as they rely primarily on DFS searches.
\end{proof}
Overall, \cref{alg:KLX2} provides a structural characterisation of graphs with $\KLX = 2$. Such a structure is presented in \cref{fig:chain_of_4_cactus_necklaces}, where short back paths are omitted and the ladder pattern is made explicit.

\begin{restatable}{proposition}{algoKLXfull}
    \label{th:algoKLX2full}
    Based on \cref{cor:charKLX2}, there exists an algorithm that accepts a graph $G$ if and only if $\KLX(G) \leq 2$. Moreover, it runs in time  $\mathcal{O}(|V| + |E|)$.
\end{restatable}

\begin{proof}[Proof overview.]
    We describe the main idea of the algorithm.

    Our goal is to verify the existence of a DFS tree $T$ satisfying \cref{cor:charKLX2}. We begin by decomposing $G$ into its block-cut tree $\mathcal{T}$. We then check that each biconnected component $B$ satisfies $\KLX(B) \leq 2$ using \cref{alg:KLX2}. As established in the proof of \cref{prop:algoKLX2}, passing \cref{alg:KLX2} ensures the existence, for each $B$, of a DFS tree $T_B$ satisfying \cref{th:biconnect_KLX2}. We further verify that the articulation points of $B$ can be embedded in such a tree (possibly after a slight modification) so as to ``locally'' satisfy the constraints stated in \cref{cor:charKLX2}.
    
    The remaining step is to combine the local DFS trees $T_B$ into a global DFS tree $T$ of $G$. 
    To this end, we determine, for each block $B$, the set of feasible root positions, and orient the edges of the block-cut tree $\mathcal{T}$ according to compatibility constraints between adjacent components. 
    If these constraints are consistent—i.e., they define a valid orientation of $\mathcal{T}$—then we can assemble the $T_B$ into a DFS tree $T$ satisfying \cref{cor:charKLX2}, and hence $\KLX(G) \leq 2$.

    All steps run in linear time. A key point is that only $\mathcal{O}(1)$ candidate DFS structures need to be considered for each block. This follows from the structural characterisation of \cref{th:biconnect_KLX2}, which restricts attention to a constant number of maximal zigzag patterns, together with the efficient handling of short back paths via articulation points.
\end{proof}

\section{An FPT algorithm for testing $\bf \textbf{KLX} \le k$.}
\label{sec:fpt-courcelle}

Having designed linear-time algorithms for deciding whether a graph has KLX less than or equal to 0, 1, or 2, we now turn to testing whether $\KLX \le k$, for a fixed integer $k$.
In particular, we present a fixed-parameter tractable (FPT) algorithm with running time $\mathcal{O}(f(k)\cdot (|V|+|E|))$ for this task.

\begin{algorithm}
\caption{An FPT algorithm to test if a given graph is of KLX at most $k$.}
\label{alg:FPT}
\begin{algorithmic}[1]
\Require A graph $G$ and an integer $k \in \mathbb{N}$.
\State Test that $\TW(G) \le k+1$. If that is not the case, return false.
\State Build $\KLX_{\le k}$ as written in the proof in \cref{KLX_MSO2}
\State Test if $G \vDash \KLX_{\le k}$, and return the result.
\end{algorithmic}
\end{algorithm}

Given a graph $G$ as input, this algorithm first ensures that the tree width of $G$ is at most $k+1$. 
By the inequality $\TW(G) \le \KLX(G)+1$ shown in the following as~\cref{th:tw_lesseq_klx}, any graph with tree width greater than $k+1$ must have $\KLX(G) \ge k+1$, and can therefore be rejected.
All graphs that pass Step 1 have bounded tree width, and hence, the following theorem by Courcelle applies. 

\begin{theorem}[Courcelle~\cite{courcelle_monadic_1990}]
If $\varphi$ is an $\MSO_2$ formula, then there exists an algorithm deciding whether a graph $G$ of tree-width at most $k$ satisfies $\varphi$ in time $O(f(\varphi,k) \cdot (|V|+|E|))$ for some computable function $f$.
\end{theorem}

It remains to prove Theorem~\ref{th:tw_lesseq_klx} and to show that the property “$\KLX(G) \le k$” can be expressed by an $\MSO_2$ formula. 

\subsection{Tree-width as a lower bound on KLX}

Given a graph $G=(V_G,E_G)$, a \emph{tree decomposition of $G$} is a labelled tree $L=(V_L,E_L)$ such that the vertices of $V_L$, called \emph{bags}, satisfy the following 3 properties:
\begin{itemize}
    \item Every bag $B\in V_L$ is a subset of $V_G$.
    \item For every edge $\{a,b\}\in E_G$, there exists a bag $B\in V_L$ such that $a,b\in B$.
    \item For every vertex $v\in V_G$, the set of bags containing $v$ induces a connected subtree of $L$. Formally, for every $v \in V_G$, we have that $L[\{B\in V_L \mid v\in B\}]$ is connected.
\end{itemize}
The \emph{width} of a tree decomposition $L$ is $\max_{B\in V_L} |B|-1$.
The \textbf{tree-width} of $G$, denoted by $\TW(G)$, is defined as the smallest width of a tree decomposition of $G$.

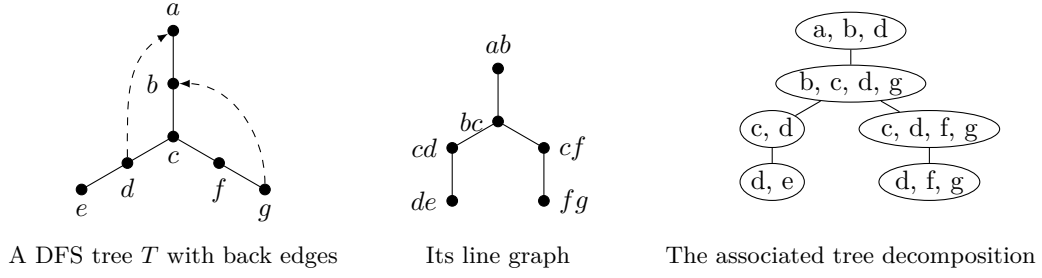
\begin{figure}[tb]
    \centering
    \begin{subfigure}{0.35\linewidth}
        \centering
        \begin{tikzpicture}
             \draw (0, 0) node[vertex] (a) [label=above:$a$]{}
            -- ++(270:0.7) node[vertex] (b) [label=left:$b$]{}
            -- ++(270:0.7) node[vertex] (c) [label=below:$c$]{}
            -- ++(210:0.7) node[vertex] (d) [label=below:$d$]{}
            -- ++(210:0.7) node[vertex] (e) [label=below:$e$]{}
            (c)
            -- ++(330:0.7) node[vertex] (f) [label=below:$f$]{}
            -- ++(330:0.7) node[vertex] (g) [label=below:$g$]{}
            ;
            \draw[dashed, -latex]
            (d) to [out=90, in=210] (a)
            ;
            \draw[dashed, -latex]
            (g) to [out=90, in=0] (b)
            ;
            \draw (0,-3) node {\small A DFS tree $T$ with back edges};

        \end{tikzpicture}
    \end{subfigure}
    \begin{subfigure}{0.25\linewidth}
        \centering
        \begin{tikzpicture}
             \draw (0, 0) node[vertex] (ab) [label=above:$ab$]{}
            -- ++(270:0.7) node[vertex] (bc) [label=left:$bc$]{}
            -- ++(210:0.7) node[vertex] (cd) [label=left:$cd$]{}
            -- ++(270:0.7) node[vertex] (de) [label=left:$de$]{}
            (bc)
            -- ++(330:0.7) node[vertex] (cf) [label=right:$cf$]{}
            -- ++(270:0.7) node[vertex] (fg) [label=right:$fg$]{}
            ;
            \draw (0,-2.5) node {\small Its line graph};
        \end{tikzpicture}
    \end{subfigure}
    \hfill
    \begin{subfigure}{0.37\linewidth}
        \centering
        \begin{tikzpicture}
        \draw (0, 0) node[ellipse, draw, inner sep=0.5] (ab) {a, b, d}
        ++(270:0.7) node[ellipse, draw, inner sep=0.5] (bc) {b, c, d, g}
        ++(210:1.2)
        node[ellipse, draw, inner sep=0.5] (cd) {c, d}
        ++(270:0.7) node[ellipse, draw, inner sep=0.5] (de) {d, e}
        (bc) 
        ++(330:1.2) node[ellipse, draw, inner sep=0.5] (cf) {c, d, f, g} 
        ++(270:0.7) node[ellipse, draw, inner sep=0.5] (fg) {d, f, g}
        ; 
        \draw (ab) -- (bc) -- (cd) -- (de)
        (bc) -- (cf) -- (fg)
        ;
        \draw (0,-3) node {\small The associated tree decomposition};
        \end{tikzpicture}
    \end{subfigure}
    
    \caption{Conversion from a DFS tree to a tree-decomposition. 
    Each bag in the tree decomposition is constructed from an edge of the DFS tree and augmented with vertices from back edges. 
    }
    \label{fig:tw_le_klx}
\end{figure}

\begin{restatable}{theorem}{thtwklx}\label{th:tw_lesseq_klx}
    For a graph $G$, $\TW(G) \le \KLX(G)+1$.
\end{restatable} 

To establish the bound claimed in~\cref{th:tw_lesseq_klx}, we convert a KLX-minimum tree $T$ of a graph $G$ into a tree decomposition $L$ whose largest bag has size at most $\KLX(G, T)+1$.
First, we construct the line graph of $T$, in which each vertex corresponds to an edge of $T$ and represents a bag. This will be the tree decomposition of $T$.
Then, for each bag $\{x, y\}$, we add the initial vertex $v$ of every back edge $(v, u)$ that either crosses or envelops $\{x, y\}$. 
For instance, in \cref{fig:tw_le_klx}, $g$ is added to $\{c,f\}$, due to $(g, b)$ crossing the edge $\{c, f\}$ and $d$ is also included because $(d, a)$ envelops $\{c,f\}$, resulting in the bag $\{c, d, f, g\}$.



\subsection{An $\bf\textbf{MSO}_2$ formula of $\bf \textbf{KLX}(G) \le k$}
\label{KLX_MSO2}

Monadic Second-Order logic ($\MSO_2$) is a logical framework that extends classical first-order logic by allowing quantification not only over individual vertices and edges, but also over sets of vertices and edges. A graph property $P(G)$ is said to be \emph{$MSO_2$-expressible} if there exists an $\MSO_2$ formula $\varphi$ such that $G \vDash \varphi$ if and only if $P(G)$ holds.


\begin{restatable}{theorem}{thKLXformula}\label{th:KLX_restatable}
    The property $\KLX_{\le k}(G) := ``\KLX(G) \le k"$ is $\MSO_2$-expressible.
\end{restatable}

The full proof is given in Appendix~\cref{proof:mso2}. 
We outline here the main ideas. 
A key difficulty is that $\MSO_2$ cannot encode a total ordering of an unbounded number of children. However, the crucial observation is that only $\KLX$ children need to be ordered:

\begin{restatable}{lemma}{lelastorder}\label{le:last_order}
    Given a graph $G$ with $\KLX(G) \le k$, there exists an ordered DFS tree $T$ of $G$ with $\KLX(T) \le k$ such that for every vertex $x$ with children $y_1,...,y_d$ ordered from left to right, only the rightmost $k$ may contain a back edge crossing $x$.
\end{restatable}

Let $k$ be the $\KLX$ value we want to test. The proof is split in four parts, each showing how to encode a concept into $\MSO_2$:

\parhead{Encoding of the DFS Tree} We define an $\MSO_2$ formula $\mathrm{DFS}(T,r,X)$ expressing that $T$ is a DFS tree rooted at $r$. $X$ is a ``labelling'' set used in the encoding to store for every vertex $v$ the information of which neighbour is its parent. Specifically, we label every edge and vertex with either $\top$ or $\bot$  according to 3 rules \textbf{R1}, \textbf{R2}, and \textbf{R3}:
\begin{itemize}
    \item \textbf{R1:} The root of $T$ is labelled $\top$, and all edges incident to it are labelled $\bot$.
    \item \textbf{R2:} For every edge $(u,v)$ of $T$, the vertex $u$ is labelled $\top$ if and only if the vertex $v$ is labelled $\bot$.
    \item \textbf{R3:} For every vertex, except the root, labelled with $x\in\{\top,\bot\}$, exactly one incident edge, namely the edge to its parent, is labelled $x$.
\end{itemize}
With this encoding, a path goes from a vertex to one of its ancestors if and only if it alternates between $\top$ and $\bot$, as depicted in \cref{fig:dfs_encoding_tree}.
We can then test that every back edge $(y,x)$ corresponds to $x$ being an ancestor of $y$. We also give a $\MSO_2$ formula $\mathrm{Parent}(x,y)$ which tests if $y$ is a parent of $x$.

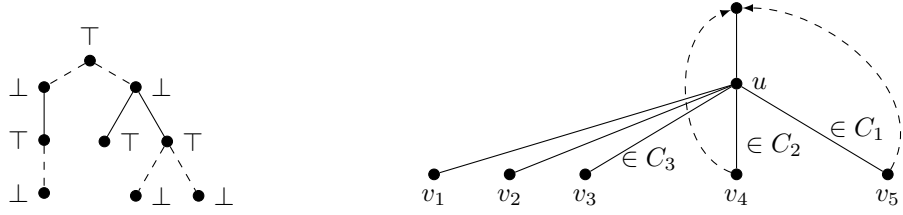
\begin{figure}[tb]
    \centering
    \hspace{2em}
    \begin{tikzpicture}
             \draw (0, 0) node[vertex] (a) [label=above:$\top$]{}
            ++(210:0.7) node[vertex] (b) [label=left:$\bot$]{}
            -- ++(270:0.7) node[vertex] (c) [label=left:$\top$]{}
            ++(270:0.7) node[vertex] (c') [label=left:$\bot$]{}
            (a)
            ++(330:0.7) node[vertex] (d) [label=right:$\bot$]{}
            -- ++(240:0.83) node[vertex] (e) [label=right:$\top$]{}
            (d)
            -- ++(300:0.83) node[vertex] (f) [label=right:$\top$]{}
            ++(240:0.83) node[vertex] (g) [label=right:$\bot$]{}
            (f)
            ++(300:0.83) node[vertex] (h) [label=right:$\bot$]{}
            ;
            \draw[dashed] (a) -- (b);
            \draw[dashed] (c) -- (c');
            \draw[dashed] (a) -- (d);
            \draw[dashed] (f) -- (g);
            \draw[dashed] (f) -- (h);
        \end{tikzpicture}
        \hfill
        \begin{tikzpicture}
             \draw (1, 0.5) node[vertex] (u) [label=right:$u$]{}
            (-3,-0.7) node[vertex] (1) [label=below:$v_1$]{}
            (-2,-0.7) node[vertex] (2) [label=below:$v_2$]{}
            (-1,-0.7) node[vertex] (3) [label=below:$v_3$]{}
            (1,-0.7) node[vertex] (4) [label=below:$v_4$]{}
            (3,-0.7) node[vertex] (5) [label=below:$v_5$]{}
            (1,1.5) node[vertex] (v) []{}
            ;
            \draw[] (u) to (1) ;
            \draw[] (u) to (2) ;
            \draw[] (u) to (v) ;
            \draw[dashed,-latex] (4) to [out=-200, in=-160] (v) ;
            \draw[dashed,-latex] (5) to [out=60, in=0] (v) ;
            \draw[] (u) to node[right, pos=0.85, xshift=0cm] {$\in C_3$} (3);
            \draw[] (u) to node[right, pos=0.7, xshift=0.cm] {$\in C_2$} (4);
            \draw[] (u) to node[right, pos=0.5, xshift=0.1cm] {$\in C_1$} (5);
        \end{tikzpicture}
        \hspace{2em}
    \caption{(Left) The encoding of a DFS tree by \textbf{R1}, \textbf{R2}, and \textbf{R3}. 
    The solid edges represent $\top$ and the dotted ones represent $\bot$. 
    (Right) The encoding of the order for $k=3$. We only store the $k$ rightmost children edges in sets $C_1$, $C_2$, and $C_3$. According to \cref{le:last_order}, if $\KLX(T) \leq k$, only $v_3$, $v_4$, and $v_5$ may contain back edges (dotted edges) crossing $u$.
    }
    \label{fig:order_k_rightmost}\label{fig:dfs_encoding_tree}
\end{figure}

\parhead{Encoding of the order} According to \cref{le:last_order}, we don't need to encode the full order of the minimal-KLX tree, but only a part of it. We encode such an order of a DFS tree $T$ using $k$ sets $C_1,...,C_k$, with $C_i$ representing the $i$-th last visited child. Each set $C_i$ is equipped with a formula $\mathrm{Choice}(C_i)$ ensuring that for all vertices $u$ at most one of its children is ``chosen''. We also require that all choices are distinct, which leads to a visiting order of the DFS tree: we first visit all non-chosen children before going into the children chosen by $C_k,...,C_1$ in this order, as illustrated in \cref{fig:order_k_rightmost}. 


We also give the formula $\mathrm{IsChoice}$ such that $\mathrm{IsChoice}(C_i,u,v)$ tests if $v$ is the child of $u$ being ``chosen'' in $C_i$. 



\parhead{Encoding of the open edges} We describe two rules \textbf{X1} and \textbf{X2} that define for every back edge $(y,x)$ a set $X_{(y,x)}$ containing all tree edges where $(y,x)$ is open. 
Specifically, \textbf{X1} refers to the edges being crossed by $(y,x)$ and \textbf{X2} refers to the edges being enveloped by it.
Formally, $X_{(y,x)}$ is the smallest set of edges satisfying:
\begin{itemize}
    \item \textbf{X1.} All tree edges in the path from $y$ to $x$ along $T$ are contained in $X_{(y,x)}$.
    \item \textbf{X2.} For every vertex $u$ strictly inside the path from $y$ to $x$, denote $c_k,...,c_1$ the choices at $u$, there is a $(c_i,u) \in X_{(y,x)}$ and for all $j < i$, $T(c_j) \subseteq X_{(y,x)}$
\end{itemize}
We then check that our rules are consistent: if $T,r,X,C_1,...,C_k$ is an encoding of a DFS tree respecting all previously mentioned rules, with $C_i$ choosing the $i$-th last child, then $X_{(y,x)}$ corresponds exactly to the set of edges where $(y,x)$ is open.


\parhead{Encoding of the KLX} Finally, we need to count how many back edges are open at the same time.
To that end, we define a series of $k$ edge sets $E_1,...,E_k \subseteq E_G$ such that every back edge $e$ finds a unique $i$ with $e \in E_i$.
Intuitively, each $E_i$ represents one “slot” for simultaneously open back edges.
If a back edge $e=(y,x)$ belongs to $E_i$, we require that $X_{(y,x)} \subseteq E_i$. This guarantees that no two back edges assigned to the same $E_i$ are open simultaneously, and thus enforces $\KLX(T,G) \le k$.

Formally, we describe three rules \textbf{E1}, \textbf{E2} and \textbf{E3} that $E_1,...,E_k$ must satisfy:
\begin{itemize}
    \item \textbf{E1.} For all back edge $e$, there is an $i$ with $e \in E_i$.
    \item \textbf{E2.} For all tree edges $e$, we have $e \in E_i$ if and only if there exists a back edge $e'$ such that $e \in X_{e'}$ and $e' \in E_i$.
    \item \textbf{E3.} For all pair of back edges $e,e'$, if $e,e' \in E_i$, then $X_{e} \cap X_{e'} = \emptyset$.
\end{itemize}
Putting everything together, we show that $G$ satisfies $\KLX(G) \le k$ if and only if there exist the variables $T, r, X, C_1, \ldots, C_k, E_1,\ldots,E_k$ satisfying all above constraints. This yields an $\MSO_2$ formula of the proposition $\KLX_{ \le k}$.

\subsection{Concluding remarks}

Algorithm~\ref{alg:FPT} serves as a proof for the following theorem. 

\begin{restatable}{theorem}{thFPTfinal}\label{th:FPTfinal}
    For all $k\in \mathbb{N}$, there exists an $\mathcal{O}(f(k)\cdot (|V|+|E|))$ algorithm that tests whether a graph $G$ has KLX number at most $k$.
\end{restatable}

Let us analyse the complexity of Algorithm~\ref{alg:FPT}. 
Testing $\TW(G) \le k+1$ (step 1) can be done in running time complexity $2^{O(k^3)} \cdot (|E|+|V|)$ \cite{bodlaender_linear-time_1996}. 
Notice that depending on how one applies Courcelle's theorem (with or without a tree-decomposition), one may not even need to run Step 1. An approximation algorithm can also be used, as just testing that the tree-width is bounded is the only goal of this step. Step 2 is constant-time as the formula $\KLX_{\le k}$ does not depend either on the graph or on the DFS tree. Due to Courcelle's Theorem, Step 3 can be done in $\mathcal{O}(f(\KLX_{\le k},k)\cdot(|V|+|E|))$. Therefore, the overall complexity is in $\mathcal{O}(g(k)\cdot (|V|+|E|))$ with $g(k) = 2^{O(k^3)}+f(\KLX_{\le k},k+1)$, linear time.

Although this algorithm as such is unlikely to serve any practical use, due to the worse-than-exponential dependency on $k$, the result suggests that further work could uncover usable versions of the algorithm.


\newpage

\bibliography{mfcs2026_KLX}
\newpage
\appendix

\section{Proofs of technical lemmas and theorems from \cref{sec:klx-leq-2}}
\label{sec:appendix-a}

\subsection{Lemmas} \label{lem:proof_lems_sec_3}
\KLXminDFStreedegreeupperbound*
\begin{proof}[Proof of \cref{lem:KLX-min-DFS-tree_degree_upperbound}.]
    If a vertex is a leaf, only one of its incident edges belongs to $T$.
    All others incident edges are back edges, and are crossing the tree edge adjacent to the leaf.
    Thus, the degree of a leaf cannot exceed $\KLX(T)+1$. 
    The biconnectivity of $G$ implies that the root is not a branching vertex (as noted in \cref{lem:root_DFST_biconnected-graph_non-branching}), allowing the same argument to be applied to the root.  
    
    Let us consider an internal vertex $u$ that is not a branch vertex. 
    In order for $u$ not to be an articulation point, some back edge $e$ must cross both the tree edges above and below~$u$. 
    If~$u$ had $2\cdot\KLX(T)+1$ or more incident edges, 
    of which only two belong to $T$, 
    one of the tree edges above or below~$u$ would be crossed by at least $\KLX(T)$ back edges. Adding the back edge $e$, this would contradict the definition of $\KLX(T)$. 
    Hence, the degree of an internal vertex that is not a branch vertex is at most $2\cdot\KLX(T)+1$.
\end{proof}


\subsection{About the characterisation of graphs with KLX 2}

\corcharKLX*

\begin{proof}[Proof of \cref{cor:charKLX2}] \label{proof:charKLX2}
    Let $G$ be a graph with $\KLX \leq 2$, and $T$ its associated KLX-minimum tree.
    Let $C$ be such a biconnected component of $G$ (if any).
    
    The previous characterisation of biconnected graph with $\KLX \leq 2$ (see~\cref{th:biconnect_KLX2}) guarantees that $T[C]$ has a specific structure of ``spanning path'', as illustrated in \cref{fig:chain_of_4_cactus_necklaces}. Let us note this restricted DFS tree as $T_C$.
    
    A key observation is that the traversal $T_C$ explores all tree edges before back paths, as established in \cref{lem:back-paths_of_KLX=2_DFST_biconnected-graph}.
    This lemma let only one non-defined order at the lowermost branching vertex. If this vertex is part of a short back path, then $T_C$ could visit it first, before the long back path.
    
    \noindent Let $x \in V_C$ be an articulation point, and let $H_x$ denote the component attached at $x$. We analyse $\KLX(H_x, x)$ depending on the position of $x$ in $T_C$:
    \begin{enumerate}
        \item \textbf{\( \bm x \) is on the main tree edge of $T_C$.}
        Then $H_x$ is explored first during the DFS, while no back edge in $C$ is open. Hence $\KLX(H_x, x) \leq 2$. After exploring $H_x$, all its back edges are closed (since $x$ is an articulation point), and the remainder of $C$ is unaffected.
        \item \textbf{\( \bm x \) is on a long back path (not last one).} By the key observation, $H_x$ is enveloped by the back edge opened by the previous back path.
        Since $G$ has $\KLX$ 2, $\KLX(H_x,x) \leq 1$. \label{itm:second}
        \item \textbf{\( \bm x \) is on a short back path (not the last being crossed by the last long back path).} 
        As such a short back path is not the last one, it is crossed by a long back path. Therefore, $H_x$ is enveloped by an edge open in a long back path and $\KLX(H_x,x) \leq 1$.
       
        \item \textbf{\(\bm x \) is on the last long back path.} there are two subcases, depending on the order of $T$.
        \begin{itemize}
            \item The last short back path crossed by the last long back path contain an articulation point $y$ such that  $\KLX(H_y,y) = 2$.
            In that case, the traversal should visit $y$ first, as the final back path must not be open before exploring $H_y$. This means that $H_x$ is enveloped, hence, $\KLX(H_x,x) \leq 1$.
            \item Otherwise, the long back path can be explored first and $\KLX(H_x,x) \leq 2$. 
        \end{itemize}
        \item \textbf{\(\bm x \) is on the last short back path being crossed by the last long back path.} This case mirrors the previous one. $\KLX(H_x,x) \leq 1$ whether there is an articulation point $y$ on the last long back path such that $\KLX(H_y,y)= 2$. Otherwise, if no such $y$ exist, $\KLX(H_x,x) \leq 2$.
        \end{enumerate}
        Therefore, connected components with KLX 2 are only articulated along tree paths and along either the
        last long or short back paths of the DFS tree $T_C$.
        All other articulation points must connect to \(\KLX\)~1 components (i.e., only cactus subgraphs). 

        This concludes the direct sense of the equivalence, the other is verified by construction.
\end{proof}

\subsection{About KLX2 algorithms}
\begin{algorithm}
\caption*{\textbf{Algorithm~\ref{alg:KLX2} (restated).} An algorithm recognising biconnected graphs with KLX $\leq 2$}
\begin{algorithmic}[1]
\State Check that there is no vertex of degree at least 5, and test whether $\KLX(G)\leq 1$. 
\State Replace every 2-high-degree cycle by a path of length~2 between the two high degree vertices. Proceed linearly, without using the ``replaced'' edges.
\State Verify that the resulting graph forms a ladder pattern.
\State Check the placement of the replaced cycles:
\begin{itemize}
    \item If the ladder pattern consists of a single cycle, accept if there is a vertex of degree 3.
    \item Otherwise, verify that there exists a zigzag motif of the ladder pattern and that all replaced cycles are correctly positioned along this zigzag motif.
\end{itemize}
\end{algorithmic}
\end{algorithm}

\begin{lemma}
    \label{lem:step3complexity}
    Step~3 of \cref{alg:KLX2} can be implemented in time $\mathcal{O}(|V| + |E|)$.
\end{lemma}
\begin{proof}[Proof of \cref{lem:step3complexity}.] 
We sketch an algorithm to verify that the graph forms a ladder pattern.
\begin{algorithmic}[1]
    \item  Check that there is no vertex of degree at least $5$.
    \item Perform a DFS traversal to identify an initial cycle of the ladder pattern (these correspond to cycles with exactly two high-degree vertices). Let $u_0$ and $v_0$ denote its high-degree vertices, and mark this cycle as processed.
    \item Iteratively, starting from $u_i$, perform local DFS explorations to find a path to $v_i$ that introduces exactly two new high-degree vertices, denoted $u_{i+1}$ and $v_{i+1}$. Mark this path as processed, and repeat until reaching a path that introduces no new high-degree vertices.
\end{algorithmic}
Steps~1 and~2 clearly run in linear time. In Step~3, each local DFS exploration traverses only previously unprocessed portions of the graph. Moreover, each edge is explored a constant number of times (namely, at most four), which ensures an overall running time of $\mathcal{O}(|V| + |E|)$.
\end{proof}

\algbicKLX*

\begin{proof}
\parhead{Correctness}
$\Longleftarrow$ Let $G$ be a biconnected graph such that $\KLX(G) \leq 2$.
\cref{th:biconnect_KLX2} provides a spanning path $T$, with back paths verifying some properties. 
\begin{itemize}
    \item \textbf{Step 1} is passed, as all vertices of $G$ are of degree at most 4 (\cref{lem:max-degree_and_KLX_lowerbound} with $k = 2$). 
    \item \textbf{Step 2.} Let $\mathcal{C}$ be a 2-high-degree cycle of $G$. By definition, $\mathcal{C}$ contains a back path of $T$.
    \begin{itemize}
        \item If this is a long back path, then it must be the first or the last one, as other internal long back path have at least 4 high degree vertices ($l_i,l'_i,l_{i+1}, l'_{i-1}$).
        By definition, it cannot be the last one, as this specific back path always cross a small back path, by construction. Therefore, $\mathcal{C}$ must involve the first long back path.
        
        In this case, it requires that this long back path does not cross any short one, otherwise, $\mathcal{C}$ would contain at least 4 high degree vertices. 
        However, in this specific situation, the first long back path could be considered as a small back path crossed by the second long back path---being a matter of where the root is located.

        \item If $\mathcal{C}$ involves a short back path, it cannot involve any other back path, as it must already contain two high degree vertices. Therefore, $\mathcal{C}$ is the short back path and its associated tree edge.
        
    \end{itemize}
    Therefore, there exists a spanning path $T$, such that step 2 removes all short back paths.
    \item \textbf{Step 3.} Before this step, only long back paths remain. They are forming a ladder pattern, as the long back paths are a sequence of cycles respecting the definition, by \cref{eq:KLX2_braids}. Therefore, $G$ passes step 3 of the algorithm.
    
    If there is more than one cycle in this ladder pattern, one can remark that the two possible zigzag motifs of this ladder pattern either include $\{(l_i,l'_i)\}$ or $\{(l'_{i-1},l_{i+1})\}$, for $1 < i < p$.
    Additionally, they include two other extra paths either involving $l'_1$ and $l_p$ or $l_2$ and $l'_{p-1}$
    \item \textbf{Step 4.} 
    \begin{itemize}
        \item 
    If there is more than one long back path, all short back paths crossed either by the first long back paths, are located exclusively between $r$ and $l_2$, by \cref{eq:KLX2_short_bps_first}. Similarly, the last ones are located exclusively between $l'_{p-1}$ and $l'_p$.
    The remaining others are located on the paths $(l'_{i-1}, l_{i+1})$, for $1 < i < p$, by \cref{eq:KLX2_short_bps}. All together, these paths form the set $\{(r,l_2), (l'_2,l_4), \ldots (l'_{p-2}, l_{p}),(l'_{p-1}, l'_p)\}$, which is a zigzag motif of the ladder pattern, as noticed previously.
    Therefore, $G$ passes step 4 the algorithm.
    \item If there is only one long back path, the root has degree 3. Hence, $G$ passes step 4 the algorithm.
\end{itemize}
\end{itemize}

$\Longrightarrow$
Reciprocally, let $G$ be a graph passing \cref{alg:KLX2}. Let us build a DFS tree $T$ on $G$ such that $\KLX(T)=2$, using the characterisation proved in \cref{th:biconnect_KLX2}.

Let $\mathcal{C}_1, \ldots,\mathcal{C}_p$ be the ladder pattern detected during step 3.
First say that $p\geq2$
Let $(u_1,d_1), \ldots,(u_{p-1},d_{d-1})$ be the paths common to each neighbouring cycles in this ladder pattern, such that all $u$ (respectively all $d$) are on the same side. Specifically, all paths from $u_i$ to $u_{i+1}$ belong only to $\mathcal{C}_{i+1}$.
Without loss of generality, let us say that the zigzag motif recognised in step 4 is $\{(v,d_1),(u_1,u_2), (d_2, d_3), (u_3, u_4), \ldots, (\cdot,w)\}$, with $v$ and $w$ being vertices located in one of the paths made of degree 2 nodes, between $u_1$ and $d_1$, or between $u_p$ and $d_p$, respectively.

Let $v$ be the root of the DFS.
Let us consider the DFS path following the zigzag such that $r=v \prec d_1 \prec u_1 \preceq u_2 \prec d_2 \preceq d_3 \prec \cdots \preceq u_{p-1} \prec d_{p-1} \prec w$ (if $p$ is odd, without loss of generality). This DFS path implies $p$ long back paths as $[r,u_1], [d_1,d_2], [u_2, u_3], \ldots$ These long back paths verify \cref{eq:KLX2_braids}, by construction.

According to \cref{th:biconnect_KLX2}, all other back edges should be short edges. Let us show that they verify their respecting equations.
First notice that all 2-high-degree cycles must be located in series. By construction during step 2 of the algorithm, the 2-high-degree cycles cannot share an edge, as this step is done sequentially, without considering replaced edges any more.
By contradiction, let us imagine that two of them share two vertices. 
Such vertices are then of degree at least 5, as they are connected to the two 2-high-degree cycles plus the tree path. This is a contradiction of the step 1 of the algorithm.
Therefore, 2-high-degree cycles can only share a single vertex, and lie in series.

On each location of the zigzag, there exists an order on small back paths, such that equations (\ref{eq:KLX2_short_bps} and (\ref{eq:KLX2_short_bps_first}) are verified by the previous construction of long back paths.

If step 3 detects only one cycle, $G$ must contain vertex $u$ with degree 3, by step 4.
As stated previously, all 2-high-degree cycles must lie on series along the single cycle detected during step 3. As $u$ is of degree 3, $u$ is a high-degree vertex of a single 2-high-degree cycle. Let $v$ be other high-degree vertex of this 2-high-degree cycle, and $(w,u)$ the edge from $u$ not being part of this 2-high-degree cycle.
Consider the DFS path $u \prec v \preceq \dots \preceq w$. By construction, this DFS path verifies \cref{eq:KLX2_short_bps_first}, with $l_1 = u$, and $l_1' = w$, as $u \not= w$.
This concludes the proof that all graph $G$ passing \cref{alg:KLX2} is of $\KLX \leq 2$, by \cref{th:biconnect_KLX2}.

\parhead{Running time analysis}
All steps have time complexity $\mathcal{O}(|V| + |E|)$:
    \begin{itemize}
        \item \textbf{Step 1} is a linear assertion on the vertices.
        \item \textbf{Step 2} can be done in  $\mathcal{O}(|V| + |E|)$ via DFS search.
        \item \textbf{Step 3} can be done in  $\mathcal{O}(|V| + |E|)$ according to \cref{lem:step3complexity}. 
        \item \textbf{Step 4}. By choosing the vertex $v$ of the zigzag motif as a direct neighbour of $u_1$, or $d_1$. Similarly, we chose $w$ as a direct neighbour of $d_p$ or $u_p$. This choice maximises the zigzag pattern. Therefore, there are only 8 maximal zigzag patterns to verify during step 4. Finding these 8 zigzag patterns in linear time (doable during step 3). All checks of the 2-high-degree cycles positions are done in linear time.
    \end{itemize}
\end{proof}

\algoKLXfull*

\begin{proof}[Proof of \cref{th:algoKLX2full}]~

    \parhead{Description of the algorithm} 
    \begin{algorithmic}[1]
        \State Decompose $G$ into its block-cut tree $\mathcal{T}$ (whose nodes correspond to biconnected components $C$).
        \item For each biconnected component $C$, test whether $\KLX(C) \leq 2$ using \cref{alg:KLX2}. If not, reject. 
        Moreover, contract each component with $\KLX \leq 1$ into a single vertex, so that $\mathcal{T}$ becomes a tree of components with $\KLX = 2$ (while preserving articulation structure/position). 
        \item For each remaining component $C$, determine whether there exists a DFS tree $T_C$ satisfying locally \cref{cor:charKLX2}; otherwise, reject. In such a tree, every articulation vertex must lie either on the main tree path or on one of the extremal back paths. This induces a set $R_C \subseteq V_C$ of feasible roots for such $T_C$.
        \State For each edge of $\mathcal{T}$ connecting components $C_1$ and $C_2$ through an articulation vertex $v$, enforce compatibility:
        \begin{itemize}
            \item If the articulation vertex does not belong to $R_{C_1}$ nor to $R_{C_2}$, reject.            
            \item If it belongs to $R_{C_2}$ but not to $R_{C_1}$, then $C_2$ must be rooted at this vertex; orient the edge of $\mathcal{T}$ from $C_1$ to $C_2$.
            \item If it belongs to $R_{C_1}$ but not to $R_{C_2}$, orient the edge from $C_2$ to $C_1$.
            \item If it belongs to both $R_{C_1}$ and $R_{C_2}$, leave the edge unoriented.
        \end{itemize}
        \State If these constraints are globally consistent (i.e., define a partial orientation of $\mathcal{T}$ with no conflict), accept; otherwise, reject.        
    \end{algorithmic}

    \parhead{Correctness} Assume first that $G$ satisfies $\KLX(G) \leq 2$. Let $T$ be a KLX-minimum DFS tree of $G$.
    Its restriction to each block $C$, denoted $T[C]$, is a DFS tree satisfying \cref{cor:charKLX2}. Hence $G$ passes Steps~2 and~3.
    Moreover, the tree $T$ induces a consistent orientation of $\mathcal{T}$. We remark here that step 4 does not orient any edge of $\mathcal{T}$ in an opposing direction.  
    Therefore, $G$ passes steps 4 and 5, and, hence, the algorithm.\\[1em]
    Conversely, assume that $G$ is accepted by the algorithm. We construct a DFS tree $T$ of $G$ satisfying \cref{cor:charKLX2} by induction. At each step of the construction $T_i$, we prove that $T_i$ is a DFS tree of $G[E_{T_i}]$ satisfying \cref{cor:charKLX2}.
    
    Initialisation. We first orient the remaining unoriented edges of $\mathcal{T}$ arbitrarily so that $\mathcal{T}$ becomes a rooted tree. There exists such an orientation as $G$ passes step~5. Let $C_r$ be the root component of $\mathcal{T}$. By step~3, there exists a DFS tree $T_{C_r}$ or $C_r$ satisfying \cref{cor:charKLX2}, and we set $T_0 := T_{C_r}$.

    Induction. We then traverse $\mathcal{T}$ in a depth-first manner. Suppose $T_i$ has been constructed for a subtree of $\mathcal{T}$, and let $C'$ be the next component to attach, with parent $C$. 
    By construction, the orientation constraints of step~4 ensure that there exists a DFS tree $T_{C'}$ rooted at the articulation vertex shared with $C$ and satisfying locally \cref{cor:charKLX2}. We attach $T_{C'}$ to $T_i$ at this vertex, obtaining a DFS tree $T_{i+1}$. Note that this ``attach'' in $G$ may be a path, as we contracted them at step~1. 
    The induction hypothesis along with the construction and the property of $T_C$ guarantee that $T_{i+1}$ is a DFS tree of $G[E_{T_{i+1}}]$, and satisfies \cref{cor:charKLX2}.

    After processing all components of $\KLX = 2$, the remaining contracted components (with $\KLX \leq 1$) are expanded by arbitrary DFS traversals from their articulation points. This yields a DFS tree $T$ of $G$ satisfying \cref{cor:charKLX2}, and thus $\KLX(G) \leq 2$.
    \\[1em]
    \parhead{Running time analysis} Step~1 runs in linear time using the algorithm of Hopcroft and Tarjan~\cite{hopcroft_algorithm_1973}. Step~2 runs in $\mathcal{O}(|V|+|E|)$ by \cref{prop:algoKLX2}. Steps~4 and~5 are linear in the size of $\mathcal{T}$.

    Step~3 is the most delicate. The key observation is that only $\mathcal{O}(1)$ candidate DFS structures need to be considered for each component. Indeed, by \cref{alg:KLX2} and the proof of \cref{prop:algoKLX2}, there are at most a constant number (namely, 8) of maximal zigzag motifs. Once such a structure is fixed, the placement of short back paths can be determined in linear time using the articulation points and \cref{cor:charKLX2}. Moreover, verifying locally that the constraints of \cref{cor:charKLX2} are satisfied can also be done in linear time.
    
    For a given component $C$, the set $R_C$ can therefore be computed in linear time. More precisely, $R_C$ consists of the vertices lying on segments that extend from the endpoint of the zigzag motif up to either the first articulation point leading to another component with KLX~2, or the first branching vertex of $T_C$.
    Overall, Step~3 can be implemented in time $\mathcal{O}(|V| + |E|)$. Consequently, the entire algorithm runs in linear time.
\end{proof}

\section{Proofs of technical lemmas from \cref{sec:fpt-courcelle}}

\subsection{Link with tree-width}

\thtwklx*

\begin{proof}  
For this, we transform the ordered DFS into an edge-tree, somewhat reminiscent of a line graph. Let $G$ be a graph such that $\KLX(G)=k$, and let $T$ be a directed DFS such that $\KLX(T)\leq k$ and denote $v = \mathcal{C}(G)$ its contour. Consider the \emph{line graph} of $T$, defined as the graph $T'$ whose vertices are the edges of $T$, $V_{T'}=E_T$ and where two edges are adjacent if one is the direct child of the other, that is
$E_{T'}=\bigl\{\{(x,y),(y,z)\} : (x,y)\in E_T \land (y,z) \in E_T\bigr\}$.

Recall how in a back edge $(x,y)$, $x \prec y$. For an edge $e\in T_E$, denote by $\mathcal{O}^{\mathrm{deep}}(e):=\{x : (x,y) \in \mathcal{O}(e)\}$ the set of all the ``deeper sides'' of the back edges that are open when visiting $e$. This will define our bags, and the intuition is that we keep the deeper side of the back-edges in the bag until we reach the second endpoint of the backward edge.

For every vertex $e=\{u,v\}$ of $T'$, define the bag $B_e:=\mathcal{O}^{\mathrm{deep}}(e)\cup \{u,v\}$.
Therefore,$|B_e|\leq |\mathcal{O}(e)|+2\leq \KLX(G)+2$. We now show that $T'$ alongside the bags $(B_e)_e$ is a tree decomposition of $G$.
\begin{description}
    \item[Property 1.] By definition, every bag is a subset of $V_G$.
    \item[Property 2.] For every edge $e\in E_G$, there are two cases:
    \begin{itemize}
        \item If $e\in E_T$, then the bag $B_e$ contains both endpoints $x$ and $y$.
        \item If $e=\{x,y\}\notin E_T$, then it is a back edge. Suppose $x\prec y$, and consider $(c,y)$ the last edge of the path $xTy$. Since $c\neq x$ because $e\notin E_T$, we have $x \in \mathcal{O}^{\mathrm{deep}}(\{c,y\})$, and therefore $x,y\in B_{\{c,y\}}$.
    \end{itemize}
    \item[Property 3.] Remark that for all $u \in V_G$, the list of the moments $t$ in the contour such that $u\in \mathcal{O}^{\mathrm{deep}}((\tau_t, \tau_{t+1}))$ is an interval: it is a union of intervals all starting at the same point. That is, a union of intervals of the form $[\![a; b_i]\!]$ with $a$ the last occurrence of $u$ in $\tau$, common to all intervals, and a serie of $(b_i)_i$. Since an interval of the contour is necessary a connected subtree of $T_E$, this gives us the property.
        
        
\end{description}
Therefore, $\TW(G)\leq \max_e |B_e|-1 \leq \KLX(G)+1$. The inequality is strict in some cases.
\end{proof}

\subsection{KLX is $\MSO_2$-expressible} \label{proof:mso2}

In this subsection, let $k \in \mathbb{N}$ be the KLX number we want to test. Given a tree $T = (V_T,E_T)$ of $G = (V_G,E_G)$, we denote by $\overline{E_T} = \mathrm{Compl}_{E_G}(E_T) = E_G \setminus E_T$ the complement of $E_T$ in $E_G$ (and not in $V_G^2$). We will show that testing if the KLX of a graph is $\le k$ is a $\MSO_2$-expressible property.

\subsubsection{Encoding of a rooted DFS Tree}

The goal of this section is to construct a formula $\mathrm{DFS}$ that encode as a $\MSO_2$ formula the fact that a set of edges $T$ is a rooted DFS. Testing whether a tree is a DFS tree has already been shown to be $\MSO_1$-expressible. For instance,~\cite{sam_parameterized_2023} showed that the class of trees admitting a DFS of height less than $k$ is $\MSO_1$-expressible. We will have to use a different encoding that is more powerful for the next few steps.

Given a spanning tree $T$ of $G$, let us label every \textbf{edge and vertex} with either $\top$ or $\bot$ according to the following three rules:
\begin{itemize}
    \item \textbf{R1:} One vertex, namely the root, is labelled $\top$, and all edges incident to it are labelled $\bot$.
    \item \textbf{R2:} For every edge $(u,v)\in T_E$, the vertex $u$ is labelled $\top$ if and only if the vertex $v$ is labelled $\bot$.
    \item \textbf{R3:} For every vertex, except the root, labelled with $x\in\{\top,\bot\}$, exactly one incident edge, namely the edge to its parent, is labelled $x$.
\end{itemize}

Rules \textbf{R1} and \textbf{R2} imply that a vertex $u$ is labelled $\top$ if and only if it is at even depth from the root.


\begin{restatable}{proposition}{FptDfsEncoding}
    \label{prop:fpt_dfs_encoding}
    Given two vertices $u,v\in V_G$, the path from $u$ to $v$ is an alternating sequence of $\top/\bot$ labels in $T$ if and only if $u\preceq v$ or $v\preceq u$.
\end{restatable}

\begin{proof}
We show by induction on the depth of a vertex $u$ that $u$ is labelled $\top$ if and only if the edge to its parent $p$ is the only incident edge labelled $\top$.

\begin{description}
    \item[Initialisation.] vertices at depth $1$ are all labelled $\bot$, and by \textbf{R1} all edges incident to the root are labelled $\bot$. Hence, for a vertex $u$ labelled $\top$, its parent is labelled $\bot$.
    \item[Inductive step.] Let $u$ be a vertex and $p$ its parent. We only treat the case where $u$ is labelled $\top$, the case of $\bot$ being identical. Then $p$ is labelled $\bot$, and the only edge incident to $p$ carrying the label $\bot$ is the edge between $p$ and its parent. Therefore the edge $(u,p)$ is labelled $\top$.
\end{description}

Finally, given two vertices $u,v$, we have $u\preceq v$ or $v\preceq u$ if and only if the path does not go from one child branch to another, which is equivalent to saying that no two consecutive labels are equal.
\end{proof}

\begin{corollary}
A tree $T$ is a DFS tree in the graph $G$ if and only if there exist a labeling of $T$ and a root vertex $r$ satisfying rules \textbf{R1}, \textbf{R2}, and \textbf{R3}, such that for every edge $(u,v)\notin T_E$, the path from $u$ to $v$ is an alternating $\top/\bot$ sequence.
\end{corollary}

The idea is now to encode a rooted DFS $T$ by means of a set $X$ containing exactly the vertices and edges labelled $\top$, the others being labelled $\bot$.

\begin{theorem}
The predicate $\mathrm{DFS}(T,r,X)$, which tests whether $T$ a set of edges is a DFS with root $r$ and with the set $X$ describes a correct labeling of $T$ satisfying \textbf{R1}, \textbf{R2}, and \textbf{R3}, is $\MSO_2$-definable.
\end{theorem}

\begin{proof}
We define formulas $\varphi_1$, $\varphi_2$, $\varphi_3$ checking respectively \textbf{R1}, \textbf{R2}, and \textbf{R3}, and a formula $\psi$ expressing that every non-tree edge induces an alternating $\top/\bot$ path (some free variables like $T, r$ or $X$ are not always included in the argument of formula for clarity):
\begin{ceqn}
\[
\begin{aligned}
\mathrm{DFS}(T,r,X) \;:=\;& \mathrm{IsTree}(T) \land r\in X \land \varphi_1 \land \varphi_2 \land \varphi_3 \land \psi,\\
\varphi_1 \;:=\;& \forall e\in E_G,\ \forall v,\ (r,v)=e \Rightarrow e\notin X,\\
\varphi_2 \;:=\;& \forall (u,v)\in E_T,\ u\in X \leftrightarrow v\notin X,\\
\varphi_3 \;:=\;& \forall v\in V_G,\ v\neq r \Rightarrow \exists (v,n)\in E_T,\ \forall (v,u)\in T_E,\\
&\qquad \bigl((v,n)\in X \leftrightarrow v\in X\bigr)
\land
\bigl(((v,u)\in X \leftrightarrow v\in X)\Rightarrow u=n\bigr),\\
\mathrm{Parent}(c,p) \;:=\;& \exists (c,p)\in E_T,\ c\in X \leftrightarrow (c,p)\in X,\\
\mathrm{IsInPath}(x,u,v) \;:=\;& \exists P \subseteq V_T, u \in P \land \forall w\in P,w=v \lor (\exists p, \mathrm{Parent}(w,p) \land p \in P) \land x \in P,\\
\mathrm{IsAncestor}(u, v) \;:=\;& \mathrm{IsInPath}(v,u,r)\\
\psi \;:=\;& \forall (u,v)\in E_G \setminus E_T,\ \mathrm{IsAncestor}(u, v).
\end{aligned}
\]
\end{ceqn}
\end{proof}

Notice that we omitted for clarity to include the formula for $\mathrm{IsTree}(T)$ that test if $T$ is indeed a spanning tree. This can easily be done with the "unique path" characterisation of a spanning tree, see \cite{sam_parameterized_2023}.

\subsubsection{Encoding a choice}

Given $T,r,X$ the encoding of a rooted $\mathrm{DFS}$ tree, we present a way to encode a choice for every vertex $u \in V_T$ of a single children (or none).

Let us label every edges of $T$ with either $\top$ or $\bot$ according to the following rule 
\textbf{RC}: For every vertex, at most one incident edge other than the parent is $\top$.

This is the one we choose. We can then easely find if $x$ is the choosen children of $y$ by testing if ${x,y}$ is marked and $\mathrm{Parent}(x,y)$, as for every vertex, at most one $y$ is such that $\{x,y\}$ is marked and $\mathrm{Parent}(y,x)$ is true.

We can then encode a choice of children as an edge set $C$ that satisfy the formula $\mathrm{Choice}(C)$ and we can retrieve the choice with the formula $\mathrm{IsChoice}(C,u,c)$ as follow:
$$
\begin{aligned}
\mathrm{Choice}(C) \;:=\; &\forall x \in V_T, \lnot\exists y,y' \in V_T, y \neq y' \land \mathrm{Parent}(y,x) \land\\ &\mathrm{Parent}(y',x) \land (x,y), (x,y')\in C) \\
\mathrm{IsChoice}(C,u,c) \;:=\; &\{c,u\} \in C \land \mathrm{Parent}(c,u)
\end{aligned}
$$

We will now consider $T, r, X, C_1,...,C_k$ an encoding of a rooted DFS represented by $T, r, X$ where for every vertex $u$ we selected at most $k$ children in an order, representing by $C_1,...,C_k$, all respecting \textbf{RC}, and \textbf{are all disjoints}. They will represent the last $k$ (at most) children visited of every vertex, with $C_i$ representing the last $i$-th children. 

\subsubsection{Encoding of the open edges}

Recall from the static characterization that in an ordered DFS $T$, a back edge $(x,y) \in \mathcal{O}(e)$ if and only if either $e$ is in the path $xTy$, or there exists $v \in xTy \setminus \{y\}$ with children $c_1$ preceding $c_2$, $x\in T(c_1)$ and $e \in T(c_2)$. We will encode this as a $\MSO_2$ formula.

For $(u,v) \in \overline{E_T}$, define $X_{(u,v)}$ to be the set of inner edges of the open edge $(u,v)$ as the smallest set satisfying:
\begin{itemize}
    \item \textbf{X1.} All tree edges in the path from $y$ to $x$ along $T$ are contained in $X_{(y,x)}$.
    \item \textbf{X2.} For every vertex $u$ strictly inside the path from $y$ to $x$, denote $c_k,...,c_1$ the choices at $u$, there is a $(c_i,u) \in X_{(y,x)}$ and for all $j < i$, $T(c_j) \subseteq X_{(y,x)}$
\end{itemize}

\lelastorder*


\begin{proof}
    Let $T$ be an ordered DFS tree of $\KLX \le k$ We will keep the same tree but change the order of the children in order to satisfy the property. Let $u\in V_T$ be a non-root vertex of children $y_1,..,y_d$ and of parent $p$. We say that a children $y_i$ is \emph{interesting} if $\mathcal{O}((y_i,x)) \cap \mathcal{O}((x,p)) \neq \emptyset$. There are at most $k$ interesting children, as otherwise $|\mathcal{O}((x,p))| > k$. Then, we know that for every non-interesting children, $\forall (a,b) \in \mathcal{O}((y_i,x))$ with $b \prec a$, we have $b = x$. Therefore those children can commute place with all the other ones as they will have open 0 back edge once $x$ is next seen. We can therefore push all non-interesting children on the left and keep the interesting children on the right-side, with the same order.

    Finally, do this for every non-root vertex in $T$ to get your $T'$. Since for the root all children are non-interesting as there is no parent, any order will suffice.
\end{proof}

We will now always assume that a DFS ordered tree $T$ with $\KLX(T) \le k$ respect the previous property.

\begin{corollary}
    If $T,r,X,C_1,...,C_k$ correspond to the encoding of a DFS ordered tree with $\KLX \le k$, with $C_1,...,C_k$ the choice of the $k$ right-most children of every vertex, then $X_{(x,y)} = \{(x,y) \in E_T \space | \space (x,y)\in \mathcal{O}((x,y)) \} $
\end{corollary}

\begin{proof}
    This is exactly the static characterization of $\mathcal{O}(e)$.
\end{proof}

This gives us the following formulas for $e \in X_{(x,y)}$. In this, for ease of comprehension "$e \in X_{(x,y)}$" can be seen as a predicate in of itself. $\vartheta_1, \vartheta_2$ represent the edges added in \textbf{X1}, \textbf{X2} respectively, and $\gamma$ test the condition in \textbf{X2} that for every vertex strictly inside the path, there is a $(c_i,u)$ being chosen. 
$$
\begin{aligned}
(a,b) \in X_{(x,y)} \;:=\; &\vartheta_1 \lor \vartheta_2 \\
\vartheta_1 \;:=\; & \mathrm{IsInPath}(e,x,y)\\
\vartheta_2 \;:=\; &\exists (c,u) \in E_T, \mathrm{IsInPath}((c,u),x,y) \land u \neq y \land\\
&\bigvee_{1\le i\le k} \big( \mathrm{IsChoice}(C_i,u,c) \land \bigvee_{j<i\le k} \exists c', \mathrm{IsChoice}(C_j,u,c') \land 
\mathrm{IsInPath}((c',u), a,r) \big) \\
\gamma \;:=\; &\forall (u,v) \in E_T, \mathrm{IsInPath}((u,v),x,y) \rightarrow \bigvee_{1\le i \le k} \mathrm{IsChoice} (C_i,v,u)
\end{aligned}
$$

\subsubsection{$\KLX_{\le k}$ is $\MSO_2$-expressible}

Consider a serie of sets $E_1, \dots ,E_k \subseteq E_G$ satisfying the following additional 3 conditions:
\begin{itemize}
    \item \textbf{E1.} For all $e \in \overline{E_T}$, there is an $i$ with $e \in E_i$.
    \item \textbf{E2.} For all $(u,v) \in E_T$, we have $(u,v) \in E_i$ if and only if there exists $(x,y) \in \overline{E_T}$ such that $(u,v) \in X_{(x,y)}$ and $(x,y) \in E_i$.
    
    \item \textbf{E3.} For all $(u,v),(u',v') \in \overline{E_T}$, if $(u,v),(u',v') \in E_i$, then $X_{(u,v)} \cap X_{(u',v')} = \emptyset$.
\end{itemize}

\begin{proposition}
Given a graph $G$, we have $\KLX(G) \leq k$ if and only if there exist an encoding $T,r,X,C_1,...,C_k$ of a DFS ordered tree and $k$ sets $E_1,\dots,E_k$ satisfying all the stated rules above.
\end{proposition}

\parhead{Intuition} The idea is that for every back edge $(u,v) \in E_G \setminus E_T$, all moments at which this edge is open $X_{(u,v)}$, that is, all traversed edges are gathered into the same set $E_i$, with no overlap for two distincts $E_i, E_j$. At the end, every $E_i$ will be a collection of circuit of the DFS, all disjoints and each of the form $X_{e} \cup \{e\}$.

\begin{proof}
We prove both implications. Denote $p(x)$ the parent of a vertex $x$ when the context is clear.

\medskip
\noindent\textbf{($\Rightarrow$)} Let $G$ be a graph such that $\KLX(G) \le k$, and let $T$ be the ordered DFS tree with contour $\tau$. Encode $T$ as previously described to obtain $T, r, X, C_1, \dots, C_k$. We define edge sets $(E_i)_{i \le k}$ incrementally by following the contour with the induction hypothesis that $|\{i : (\tau_t,\tau_{t+1}) \in E_i\}| = |\mathcal{O}(t+1)|$ by induction on $t< |\tau(T)|$: 
\begin{itemize}
    \item Initially, for all $i$, $E_i = \emptyset$.
    \item For a given transition of the DFS $(\tau_t, \tau_{t+1})$ that is a downward edge with $\tau_{t} \prec \tau_{t+1}$, it means $\tau_{t+1}$ is a vertex we have never seen. Then keep the same $(E_i)_i$: we will see $\tau_{t+1}$ again when coming back so we don't open any new edges.
    \item For a given transition of the DFS $(\tau_t, \tau_{t+1})=(c,x)$ that is an upward edge  with $x\prec c$, for every backward edge of the form $(x,y)$, there is some $j \le k$ such that $(x,p(x)) \notin E_j$ (otherwise we would have more than $k$ open edges at the same time on $(x,p(x))$). Add the moments when $(x,t)$ is open to $E_j$ (by the lemma, these are $X_{(x,t)}$), and add $(x,t)$ to $E_j$.
\end{itemize}

Let us verify that all rules are satisfied:
\begin{itemize}
    \item \textbf{R1}, \textbf{R2} and \textbf{R2} hold because of the previous results.
    \item \textbf{E1}: By definition of the $E_i$, for every backward edge $(u,v)$, we added it to an edge set when considering the transition at time $t_u$ in $\tau$.
    \item \textbf{E2}: For a given backward edge $(x,y)$, we added to the set $E_i$ only the edges of $E_T$ where $(x,y)$ is open, that is $X_{(x,y)}$. Since we only add an edge of $E_T$ in that case and start with the empty set for $E_i$, by induction these are the only ones that remain.
    \item \textbf{E3}: By contradiction, let $(u,v)$ and $(u',v')$ be two back edges, with $v \prec u$ and $v' \prec u'$, and suppose that $X_{(u,v)} \cap X_{(u',v')} \neq \emptyset$. Then there is an edge $(a,b) \in X_{(u,v)} \cap X_{(u',v')}$ that has been added twice to a set $E_i$. Suppose that $(u,v)$ appears before $(u',v')$ in $\tau$ (the other case is symmetric). Then, at step $t_{u'}$, when considering the edge $(u',v')$, we have $(u',p(u')) \notin E_i$, since this is the necessary condition for it to be added to $E_i$. But we know that $X_{(u,v)}$ represents a continuous interval of the DFS walk, but both $(u,p(u)) \in X_{(u,v)}$ (comes before $(u',p(u'))$) and $(a,b) \in X_{(u,v)}$ (comes after $(u',p(u'))$), so $(u',p(u'))$ must belong to $X_{(u,v)}$ and therefore cannot belong to $E_i$, a contradiction.
\end{itemize}

\medskip
\noindent\textbf{($\Leftarrow$)} Assume that such sets satisfying all the stated rules exist. We show that there exists a DFS ordered tree $T$ of contour such that, for every transition $(u,v)$ in the walk, the number of open edges at that point is at most $k$. Take the same $T$ for the DFS tree (the previous results assure that $T$ is indeed a DFS tree of root $r$).

Given a vertex $u$, let $y_1,\dots,y_n$ be the children of $u$. Denote $c_1,...,c_k$ the choosen children of $u$, defined by $C_1,...,C_k$. We define the order of $T$ at $u$ as follows: first visit every children of $u$ that is not among the choosen ones, in any arbitrary order, and then visit the children $c_1,c_{2},\dots,c_k$ in that order.

To prove the bound, we show that for every transition $(u,v)$ in the DFS,
$|\mathcal{O}({\{u,v\}})| = |\{i : (u,v) \in E_i\}| \leq k$.
Indeed, denoting $[k] = \{1,\dots,k\}$, we have
\[
\begin{aligned}
\left|\mathcal{O}({\{u,v\}})\right|
&= \left|\left\{(a,b)\in \overline{T_E} \mid (u,v)\text{ occurs while }(a,b)\text{ is open}\right\}\right| \\
&= \left|\left\{(a,b)\in \overline{T_E} \mid (u,v)\in X_{(a,b)}\right\}\right|
&& \text{by the lemma} \\
&= \left|\left\{(a,b)\in \overline{T_E} \mid \exists i,\ (u,v)\in X_{(a,b)} \text{ and } (a,b)\in E_i\right\}\right|
&& \text{by \textbf{E1}} \\
&= \left|\left\{(a,b,i)\in \overline{T_E}\times [k] \mid (u,v)\in X_{(a,b)} \text{ and } (a,b)\in E_i\right\}\right| && \text{by \textbf{E3} {\small (only once)}} \\
&= \left|\left\{i \in [k] \mid \exists (a,b)\in \overline{T_E},\ (u,v)\in X_{(a,b)} \text{ and } (a,b)\in E_i\right\}\right| && \text{by \textbf{E3} {\small (only once)}} \\
&= \left|\left\{i \in [k] \mid (u,v)\in E_i\right\}\right| && \text{by \textbf{E2}} \\
&\leq k.
\end{aligned}
\]
This concludes the proof.
\end{proof}

\thKLXformula*

\begin{proof}
This is a corollary to the previous lemma. We just have to give the final formula. Rules \textbf{E1}, \textbf{E2} and \textbf{E2} will be represented by respectively $\varrho_1, \varrho_2, \varrho_3$. $\varsigma$ represent the fact that all $C_i$ are disjoint. The formula is (finally!):
\begin{ceqn}
\[
\begin{aligned}
\KLX_{\le k} \;:=\; &\exists T,r,X,C_1,\dots,C_k,E_1,\dots,E_k, \\
&\mathrm{DFS}(T,r,X) \land \big(\bigwedge_{i=1}^k \mathrm{Choice}(C_i)\big)\land \varrho_1 \land \varrho_2 \land \varrho_3 \land \varsigma \land \gamma\\
\varsigma \;:=\; &\forall u,c \in V_G,\bigwedge_{i=1}^k \bigwedge_{j\neq i} \lnot(\mathrm{IsChoice}(C_i,u,c) \land \mathrm{IsChoice}(C_j,u,c))\\\
\varrho_1 \;:=\; &\forall e \in \overline{E_T}, \bigvee_{i=1}^k e \in E_i \\
\varrho_2 \;:=\; &\forall e \in E_T, \bigwedge_{i=1}^k (e \in E_i) \leftrightarrow (\exists e' \in \overline{E_T}, e \in X_{e'} \land e' \in E_i) \\
\varrho_3 \;:=\; &\forall e,e' \in \overline{E_T}, \bigwedge_{i=1}^k e, e' \in E_i \Rightarrow X_{e} \cap X_{e'} = \emptyset \\
\end{aligned}
\]
\end{ceqn}
\end{proof}

\end{document}